\PassOptionsToPackage{pdftex,
pdfencoding=auto,
pdfnewwindow=true,
pdfusetitle=true,
bookmarks=true,
bookmarksnumbered=true,
bookmarksopen=true,
pdfpagemode=UseThumbs,
bookmarksopenlevel=1,
pdfpagelabels=false,
breaklinks=true,
colorlinks=true
}{hyperref}
\PassOptionsToPackage{usenames,dvipsnames,pdftex}{xcolor}
\PassOptionsToPackage{letterpaper,margin=0cm,bottom=0.1cm}{geometry}
\documentclass[prx,11pt,letterpaper,twocolumn,accepted=2023-11-27]{quantumarticle}
\pdfoutput=1
\usepackage[letterpaper,margin=0cm,bottom=0.1cm]{geometry}
\usepackage[english]{babel}
\usepackage[numbers,sort&compress]{natbib}
\usepackage{amsmath,amssymb,graphicx,enumerate,bbm,xcolor,latexsym}
\usepackage{amsthm}
\usepackage{mathtools}
\usepackage{soul}
\usepackage{comment}
\usepackage[caption=false]{subfig}
\usepackage{paralist, tabularx}
\usepackage{nicefrac}
\usepackage[T1]{fontenc} 
\usepackage[utf8]{inputenx}

\usepackage[subfigure,titles]{tocloft}
\usepackage[page,title,toc]{appendix}

\usepackage{microtype}
\microtypecontext{spacing=nonfrench}
\microtypesetup{
expansion={true,nocompatibility},
protrusion={true,nocompatibility},
activate={true,nocompatibility},
tracking=true,
kerning=true,
spacing={true}
}
\usepackage[all=normal,floats=tight,mathspacing=tight,wordspacing=tight,paragraphs=normal,tracking=tight,charwidths=tight,mathdisplays=normal,sections=normal,margins=normal]{savetrees}
\everypar=\expandafter{\the\everypar\loosness=-1 }
\definecolor{myurlcolor}{rgb}{0,0,0.7}
\definecolor{myrefcolor}{rgb}{0.8,0,0}
\definecolor{purple}{RGB}{128,0,128}
\definecolor{ultramarine}{RGB}{63, 0, 255}
\definecolor{medblue}{RGB}{0, 0, 100}
\definecolor{googleblue}{RGB}{34, 0, 204}
\definecolor{panblue}{RGB}{0,24,150}
\definecolor{carmine}{RGB}{150, 0, 24}
\definecolor{gray}{RGB}{150, 150, 150}

\usepackage[unicode=true,pdfusetitle, bookmarks=false,bookmarksnumbered=false,
bookmarksopen=false, breaklinks=false,pdfborder={0 0 0},backref=false,
colorlinks=true,
linkcolor=carmine,
citecolor=googleblue,
urlcolor=panblue,
anchorcolor=OliveGreen,
]{hyperref}

\newcommand{\rob}{\color{black}}

\newcommand{\blk}{\color{black}}

\newcommand{\ket}[1]{\left| #1 \right>}
\newcommand{\bra}[1]{\left< #1 \right|}

\newcommand{\beq}{\begin{equation}}
\newcommand{\eeq}{\end{equation}}

\usepackage{braket}

\newcommand{\LOSR}[0]{\ifmmode\textup{\upshape LOSR}\else{\textup{\upshape LOSR}}\fi}
\newcommand{\LO}[0]{\ifmmode\textup{\upshape LO}\else{\textup{\upshape LO}}\fi}
\newcommand{\LOCC}[0]{\ifmmode\textup{\upshape LOCC}\else{\textup{\upshape LOCC}}\fi}
\newcommand{\LU}[0]{\ifmmode\textup{\upshape LU}\else{\textup{\upshape LU}}\fi}

\usepackage[only,Arrownot]{stmaryrd}

\theoremstyle{plain}
\newtheorem{theo}{Theorem}

\newtheorem{prop}[theo]{Proposition}
\newtheorem{lem}[theo]{Lemma}
\newtheorem{cor}[theo]{Corollary}

\theoremstyle{definition}
\newtheorem{defn}{Definition}

\usepackage[style=american]{csquotes}
\usepackage{centernot}
\newcommand{\nimplies}{\centernot{\implies}}
\usepackage{authblk}

\begin{document}
\title{Understanding the interplay of entanglement and nonlocality: 
motivating and developing a new branch of entanglement theory}

\author{David Schmid}
\affiliation{Perimeter Institute for Theoretical Physics, 31 Caroline Street North, Waterloo, Ontario Canada N2L 2Y5}
\affiliation{Institute for Quantum Computing and Department of Physics and Astronomy, University of Waterloo, Waterloo, Ontario N2L 3G1, Canada}
\affiliation{International Centre for Theory of Quantum Technologies, University of Gda\'nsk, 80-308 Gda\'nsk, Poland}
\author{Thomas C. Fraser} 
\affiliation{Perimeter Institute for Theoretical Physics, 31 Caroline Street North, Waterloo, Ontario Canada N2L 2Y5}
\affiliation{Institute for Quantum Computing and Department of Physics and Astronomy, University of Waterloo, Waterloo, Ontario N2L 3G1, Canada}
\author{Ravi Kunjwal}
\affiliation{Centre for Quantum Information and Communication, Ecole polytechnique de Bruxelles,
	CP 165, Universit\'e libre de Bruxelles, 1050 Brussels, Belgium}
\author{Ana Bel\'en Sainz}
\affiliation{International Centre for Theory of Quantum Technologies, University of Gda\'nsk, 80-308 Gda\'nsk, Poland}
\author{ Elie Wolfe}
\affiliation{Perimeter Institute for Theoretical Physics, 31 Caroline Street North, Waterloo, Ontario Canada N2L 2Y5}
\author{Robert W. Spekkens}
\affiliation{Perimeter Institute for Theoretical Physics, 31 Caroline Street North, Waterloo, Ontario Canada N2L 2Y5}

\maketitle
\begin{abstract}
A standard approach to quantifying resources is to determine which operations on the resources are freely available, and to deduce the partial order over resources that is induced by the relation of convertibility under the free operations. If the resource of interest is the nonclassicality of the correlations embodied in a quantum state, i.e., {\em entanglement}, then the common assumption is that the appropriate choice of free operations is Local Operations and Classical Communication (LOCC). We here advocate for the study of a different choice of free operations, namely, Local Operations and Shared Randomness (LOSR), and demonstrate its utility in understanding the interplay between the entanglement of states and the nonlocality of the correlations in Bell experiments. Specifically, we show that the LOSR paradigm  (i) provides a resolution of the {\em anomalies of nonlocality}, wherein partially entangled states exhibit more nonlocality than maximally  entangled states, (ii) entails new notions of genuine multipartite entanglement and nonlocality that are free of the pathological features of the conventional notions, and (iii) makes possible a resource-theoretic account of the self-testing of entangled states which generalizes and simplifies prior results. Along the way, we derive some fundamental results concerning the necessary and sufficient conditions for convertibility between pure entangled states under LOSR and highlight some of their consequences, such as the impossibility of catalysis for bipartite pure states. The resource-theoretic perspective also clarifies why it is neither surprising nor problematic that there are mixed entangled states which do not violate any Bell inequality. Our results motivate the study of LOSR-entanglement as a new branch of entanglement theory.
\end{abstract}
\tableofcontents 

\section{Introduction}\label{sec:introduction}

The term ``entangled'' was first used only for pure states of a composite system, and meant simply that the state was not a tensor product of states of the components~\cite{Schrodinger1935}.  Thus, for pure states, entanglement is synonymous with correlation.  When the quantum information community turned its attention to mixed states, the term ``entangled'' obtained a broader meaning, aimed at capturing the {\em nonclassicality} of correlations.  Specifically, a quantum state was taken to exhibit nonclassical correlations if it could not be expressed as a mixture of product states~\cite{werner1989quantum}, in which case it was called {\em nonseparable}. Shortly thereafter, it was realized that entangled states (both pure and mixed) could be used to implement useful information-processing tasks, and they began to be studied as a resource. Because the tasks being considered at the time mainly concerned the resourcefulness of entangled states in circumstances wherein the separated parties had access to classical communication channels (for instance, their use in simulating quantum channels via the teleportation protocol~\cite{teleportation}, and in enhancing communication via the dense coding protocol~\cite{densecoding}),
it was natural to define the interconvertibility preorder of entangled states relative to {\em Local Operations and Classical Communication} (LOCC)~\cite{BBPS96}.
This choice was consistent with the previous definition of the boundary between entangled and unentangled states, since the states one can prepare freely by $\LOCC$ are precisely the separable states.

However, $\LOCC$ is not the only choice of free operations that could have been used to formalize the notion of entanglement as a resource.  Consider the set of {\em Local Operations and Shared Randomness} (LOSR), wherein the parties have access to a common source of classical randomness, but no classical channels among them.  If one chooses LOSR as the set of free operations, one also reproduces the standard definition of entangled states as nonseparable states, since the free states relative to $\LOSR$ are also the separable ones.
 The ordering induced over entangled states by LOSR, however, is different from the one induced by LOCC, even in the case of pure states, as we will show. Consequently, quantification of entanglement relative to LOSR leads to quite different results than one obtains by quantifying it relative to LOCC.

To distinguish these two notions of entanglement, we will henceforth use the terms {\em LOCC-entanglement} and {\em LOSR-entanglement}.

In this article, we advocate for the development of the theory of LOSR-entanglement. 
 We motivate  its study by  demonstrating how much light it sheds on  the interplay of entanglement and nonlocality.  Specifically, we argue that for one of the most natural ways of conceptualizing a Bell scenario,
it is LOSR-entanglement that is the relevant resource of entanglement, rather than LOCC-entanglement.  We describe many ways in which conceptual puzzles regarding the interplay of entanglement and nonlocality are resolved in this approach.   
 The notion of LOSR-entanglement was originally proposed by Buscemi~\cite{Buscemi2012LOSR}, also in the context of Bell scenarios, but no further work has been done to date on characterizing it. We hope that the arguments provided herein for its importance  will motivate researchers to turn their attention to it.

The term `box' will here be used as jargon for a  multipartite process with only classical inputs and classical outputs   which can be realized by a common source 
    (either classical or quantum) which is shared among the parties and subjected to local measurements.  In other words, a box has the structure of a Bell experiment.\footnote{Such processes are termed `common-cause boxes' in Ref.~\cite{wolfe2020quantifying}.  Note that we are here only interested in boxes that are quantumly realizable, rather than the strictly larger set of boxes that are realizable in the framework of Generalized Probabilistic Theories~\cite{barrettGPT,hardy01}.
    Note also that the term `box' is sometimes used in a manner that does not presume that the internal causal structure is that of local measurements on a common source. This is done, for instance, by authors who would prefer to make no assumptions about a box's inner workings and to rely instead on assumptions about the spatio-temporal relations among its inputs and outputs.  We discuss this alternative approach in Sec.~\ref{NetworkStructure}.
}
Formally, a box is represented by the conditional probability distribution over its classical outputs given its classical inputs. Boxes can be divided into those whose correlational properties are classical and those for which they exhibit nonclassicality, where the division is based on whether the corresponding conditional probability distribution satisfies  all the Bell inequalities or not.  These two classes are conventionally termed ``local'' and ``nonlocal''.~\footnote{\label{footnote:nolocality}Although we are following a standard convention in referring to such nonclassicality of boxes as 
  ``nonlocality'', we note that this is merely for the sake of making our article easier to read. 
The conventional terminology is actually a potential source of confusion insofar as it suggests  a commitment to a view that many (including the present authors) do not endorse, namely, that the correct explanation of Bell inequality violations
involves superluminal causes.  See Sec.~2.3.1 of Ref.~\cite{wolfe2020quantifying} for more discussion of this issue.
 Note, furthermore, that the adjective `nonlocal' has sometimes been used to delineate those quantum states that can be used to violate a Bell inequality in the Bell scenario. As we argue in Appendix~\ref{Werner},
  however, being nonlocal in this sense should not be considered a necessary condition for the correlation properties of a quantum state to be judged nonclassical. 
 In any case, in this article, we will use the term ``nonlocal'' {\em solely} as a descriptor of boxes, where we will take it to signify nonclassicality of the correlations that the box describes.  }

Many authors have argued that entanglement and nonlocality are simply {\em different kinds of resources}.  Indeed, this is a standard response to some of the puzzling features of their interplay.  
 In this article, however, we take a different point of view. 
We argue that entanglement and nonlocality quantify the same notion of resourcefulness
 for the processes to which they apply, namely, the  {\em nonclassicality} of the correlational properties of those processes.  Entanglement refers to the nonclassicality of the correlational properties of {\em quantum states}, while nonlocality refers to the nonclassicality of the correlational properties of {\em boxes}.

Furthermore, we argue that whether given correlational properties (of a state or of a box) should be deemed nonclassical depends on the network connecting the parties. 

Prior results that appeared puzzling are seen, in retrospect, to be a result of mixing together notions of nonclassicality related to different networks. 
 We demonstrate that the choice of network structure that fits best with pre-existing ideas regarding the interplay of the nonclassicality of correlations of states (entanglement) and the nonclassicality of correlations of boxes (nonlocality) is the network where the parties merely share a common source.  In such a network, the relevant notion of entanglement is LOSR-entanglement.

We now summarize the rest of the article.   

In Section~\ref{sec:unifiedRT}, we explain why a resource theory that  encompasses {\em both} entangled states and nonlocal boxes---as different {\em types} of resources of nonclassicality of correlations---must be based on a type-independent constraint defining the free operations, which is then particularized to conversion relations among specific types, such as conversions from states and boxes. 
  (Note that one has {\em no choice} but to work within such a  mixed-type resource theory,  because conversions from entangled states to nonlocal boxes are {\em precisely} the focus of any study of the interplay of entanglement and nonlocality in Bell scenarios.)   
We explain why these free operations must {\em include} all of LOSR if the objective is to quantify the nonclassicality of the correlations.

The three following sections of the article explain how a reconsideration of Bell scenarios in terms of LOSR-entanglement (rather than LOCC-entanglement) resolves some problems and clarifies many issues from the Bell literature.

In Section~\ref{sec:anomalies}, we consider various conceptual puzzles
 surrounding 
  {\em anomalies of nonlocality}~\cite{Methot2006anomaly,AnomalyExtra2005Scarani,AnomalyExtra2008Brunner,vidick2011more,AnomalyExtra2011Junge,Acin2012randomnessvsnonlocality,AnomalyExtra2014Tan,AnomalyExtra2015Augusiak,AnomalyExtra2015Fonseca,AnomalyExtra2016Bowles,AnomalyExtra2017Kabel,AnomalyExtra2018Curchod,AnomalyExtra2018Bamps,AnomalyExtra2018Chitambar,AnomalyExtra2018Lipinska,AnomalyExtra2018Barasinski},
that is, situations wherein features of nonlocal boxes
 are found to be realizable by a partially entangled state but {\em not} by a maximally entangled state.
The lesson that has until now been drawn from these  anomalies is that, in spite of prior intuitions to the contrary, there are measures of the nonlocal yield of a state (i.e., the nonlocality of boxes that can be obtained from the state) that are not monotonically related to the state's entanglement.
We show, however, that there is a more productive conclusion to be drawn, namely, that the counterintuitive features of the anomalies are best understood to be a consequence of implicitly evaluating state to box conversions relative to LOSR but state to state conversions (and thus entanglement) relative to LOCC.  
The interplay of entanglement and nonlocality becomes intuitive if one instead takes the appropriate notion of entanglement to be the one based on LOSR.
One's prior intuitions are in fact vindicated when one proceeds in this fashion.  For instance, we show that
 every measure of the nonlocal yield of a given state {\em is} a valid measure of the state's LOSR-entanglement.

In Section~\ref{sec:multipartite}, we show that by focussing on LOSR-entanglement rather than LOCC-entanglement,
 one can resolve an analogous (but not previously articulated) anomaly concerning the interconversion between genuine 3-way entangled states and genuine 3-way nonlocal boxes. The resolution highlights the fact that the notion of genuine multipartite entanglement changes when entanglement is judged relative to LOSR rather than LOCC. Furthermore, our notion of genuine multipartite entanglement does not have a pathological property that the traditional notion exhibits, namely, the failure of the closure under tensor products of the states which are {\em not} genuinely multipartite entangled~\cite{navascues2020genuine,contreras2021,MultipartiteProblem2020Luo}.
 
In Section~\ref{sec:selftesting}, we demonstrate that well-known results concerning {\em self-testing} of entangled states~\cite{mayers1998quantum,mayers2003self,vsupic2020self,Scarani2019}  can be better understood in terms of the interplay of the nonlocality of boxes and the LOSR-entanglement of states.
While our resource-theoretic approach to self-testing coincides with the usual approach for pure states and convexly extremal boxes, we show that it provides a corrective to the standard definition for the case of mixed states and convexly nonextremal boxes. In doing so, we show that both of these appear in nontrivial instances of self-testing, despite previous claims to the contrary. We also show that the clarity of our principled approach to self-testing makes it easy to resolve ambiguous cases (e.g., chiral states), and to extend self-testing to novel scenarios and even novel types of resources and novel resource theories.

In Section~\ref{sec:RTLOSR}, we derive some general results about the preorder of pure entangled states under LOSR. These results will justify some of the critical steps in our arguments, so we will be referencing forward to them throughout the text.

Section~\ref{NetworkStructure} 
explains
 why the nonclassicality of correlations, and thus entanglement, is only defined {\em relative to} a network structure among the parties. We also note that one can study the interplay of entanglement and nonlocality in a network structure incorporating classical communication between the parties by restricting attention to  boxes that have space-like separated wings.  We contrast the nature of this interplay with the one observed for the network with common sources, and we highlight what needs to be done to properly formalize such an approach.

In Section~\ref{Werner}, we discuss the consequences of our approach for what is often taken to be a surprising aspect of the interplay of entanglement and nonlocality, namely, the fact that there are mixed entangled states that cannot violate any Bell inequality. 

 Finally, in Section~\ref{sec:Discussion}, we provide a discussion of the results and future work.

\section{Nonclassicality of correlations for states and boxes}\label{sec:unifiedRT}

Understanding the interplay between the entanglement of states and the nonlocality of boxes means understanding whether particular types and measures of entanglement of states are required to realize particular types and measures of nonlocality of boxes.
 In order to do so, one must articulate precisely what operations are assumed to be freely available in converting states to boxes.
  But in addition to this, one must specify what operations are freely available in achieving conversions among boxes,
  because the convertibility relations among boxes
    determine measures of nonlocality (via order-preserving functions, i.e., monotones), and one must also specify what operations can be used to achieve conversions
   among states, because the convertibility relations among states 
   determine measures of entanglement.  Consequently, there are three choices of free operations of interest---those governing box-to-box conversions, those governing state-to-box conversions, and those governing state-to-state conversions.
 
The free operations governing each of these type-specific varieties of conversion cannot be stipulated arbitrarily.  {\em They must be understood as being induced by some type-independent constraint that is then particularized to these cases.}  It has been argued elsewhere~\cite{coecke2016mathematical,marvian2016quantify} that a given choice of the set of free operations in a resource theory is physically interesting (as opposed to being of mere mathematical interest) only if it is motivated by some restriction on physical or experimental capabilities.\footnote{In the framework for resource theories set up in Ref.~\cite{coecke2016mathematical}, the nature of the physical restriction is presumed to have some structural properties, such as the free operations being closed under parallel and serial composition.}  
Insofar as a preparation of a resource of a given type is also a kind of conversion relation, namely, from the trivial type  (no systems) to the type of the resource, the boundary between free and nonfree for every different type of resource also cannot be stipulated arbitrarily but is induced by the type-independent constraint that is then particularized.

Historically, the question of whether a given entangled state can generate a given nonlocal box has 
been interpreted as the question of whether there exists some set of quantum measurements on each wing that can be implemented on the entangled state to yield the conditional probability distribution of outcomes given settings which is associated to the nonlocal box.  (For instance, this is the case in discussions of self-testing of states by boxes, as we note in Sec.~\ref{sec:selftesting}.)  In other words, whether a given state-to-box conversion relation holds or not is traditionally evaluated relative to Local Operations (LO).
  The reason, presumably, is that state-to-box conversions have heretofore been conceptualized as analogues of a Bell experiment, wherein the choice of local measurement at one wing has traditionally been presumed to be independent of the choice at any other wing, even though this independence is not needed to derive the Bell inequalities.\footnote{The assumption that the setting variables are independent of the hidden variables, on the other hand, {\em is} needed to derive the Bell inequalities.}   
   
When one conceptualizes state-to-box conversions in a resource-theoretic way, however, it becomes apparent that this LO-based approach is untenable, as we now demonstrate.

First recall that, as we noted in the introduction, the entanglement of states and the nonlocality of boxes quantify the {\em nonclassicality} of the correlational properties of states and boxes.  For both states and boxes, the distinction between free and nonfree  is the distinction between classical and nonclassical correlations.  In the case of quantum states, this corresponds to the distinction between separable and nonseparable, while in the case of boxes, it corresponds to the distinction between satisfying all Bell inequalities and violating some Bell inequality.  
Next, note that the set of separable states and the set of Bell-inequality-satisfying boxes cannot be generated by local operations alone; they require shared randomness.

Because a preparation of a box is a special case of 
a state-to-box conversion 
where the input type is trivial, if one were to assume LO as the set of free operations for state-to-box 
conversions, one would be stipulating that the distinction between free and nonfree boxes is the distinction between uncorrelated and correlated (i.e., product and nonproduct forms), rather than the distinction between Bell-inequality-satisfying and Bell-inequality-violating.  Consequently, an LO-based approach cannot capture the classical-nonclassical distinction. 

Furthermore, if the free operations are to be independent of type, 
then if one were to take LO as the set of free operations governing state-to-box conversions, one would also have to take LO to also govern state-to-state conversions, so that the distinction between free and nonfree {\em states} would also correspond to the distinction between product and nonproduct forms, rather than the distinction between separable and nonseparable, and thus would again not capture the classical-nonclassical distinction.  One must therefore reject the historical LO-based approach to the study of state-to-box conversions.

We now articulate our preferred approach to a resource-theoretic study of the interplay of entanglement and nonlocality.  We assume that the parties are connected by a network wherein they all have access to a common source, but where there are no channels between them, so that the distinction between free and nonfree operations is the distinction between what can be achieved by a common {\em classical} source (shared randomness) and what can be achieved by a common {\em quantum} source (nonseparable states).
This is a type-independent restriction. It follows that the set of free operations governing all varieties of conversion relations, including state-to-state, state-to-box, and box-to-box, is LOSR.  Further discussion of this proposal is provided in Sec.~\ref{NetworkStructure}.
 
An alternative approach  is one wherein the network includes channels among the parties, implying that LOCC is the set of free operations for all varieties of resource conversion.   At first glance, it might seem that the latter approach cannot possibly capture
  the nonlocality of boxes, as classical communication can be used to simulate any Bell inequality violation without requiring nonclassicality.  As we note in Sec.~\ref{NetworkStructure}, however, such a conceptualization can be made consistent by restricting attention to 
   a subclass of boxes, and it may be the pertinent one for certain applications.  Nonetheless, we shall show in Sec.~\ref{NetworkStructure} that it is rather more difficult to formalize than the one we pursue here and that the interplay of entanglement and nonlocality that it implies involves a more significant departure from standard intuitions than the one based on LOSR.  This also motivates our focus on LOSR in this article.

As we argued above, one cannot leave out shared randomness when assessing resource conversions
 if the resource of interest is the 
 nonclassicality of states and boxes.  In spite of this, there are special cases of state-to-box  conversions wherein the shared randomness does not offer any additional power over LO.  This occurs if the box is convexly extremal in the set of quantumly realizable boxes.\footnote{Here, a box is said to be quantumly realizable if it can be obtained from some quantum state by some LOSR operation. Note, however, that one could equally well define a box to be quantum realizable if it can be obtained from a quantum state by an LO operation, since the shared randomness can always be provided by the quantum state.}
 Similarly, for {\em box-to-box} conversions where the output box is convexly extremal,  LO and LOSR also deliver the same verdicts about convertibility relations. In fact, there is a slightly larger set of output boxes for which LO and LOSR deliver the same verdicts for state-to-box and box-to-box conversions, namely, those that are {\em LO-equivalent} to a convexly extremal box. The result can be summarized as follows:

\begin{lem}
\begin{samepage}
    \label{convext}
    Consider the following statements about interconversion between an $n$-partite state $\rho$ and an $n$-partite box $B$
    \begin{compactenum}[(i)]
        \item  $\rho \mapsto B$ by LOSR,
        \item  $\rho \mapsto B$ by LO,
    \end{compactenum}
    and between a pair of $n$-partite boxes,  $B_0$ and $B$,
        \begin{compactenum}[(i)]
        \item[(i)$^{\prime}$]  $B_0 \mapsto B$ by LOSR,
        \item[(ii)$^{\prime}$]  $B_0 \mapsto B$ by LO.
    \end{compactenum}
The following implications hold among these conditions:
\end{samepage}
 \begin{enumerate}
 \item[(a)] If $B$ is a convexly extremal box or a convexly nonextremal box that is LO-equivalent to a convexly extremal box,
    then  (i) and (ii) are equivalent and (i)$^{\prime}$ and (ii)$^{\prime}$ are equivalent. 
\item[(b)] If $B$ is an arbitrary convexly nonextremal box 
 then although it is still the case that (ii) $\implies$ (i) and (ii)$^{\prime}$ $\implies$ (i)$^{\prime}$, it can happen that (i) $\nimplies$ (ii) and it can happen that (i)$^{\prime}$ $\nimplies$ (ii)$^{\prime}$.
 \end{enumerate} 
 Here, convex extremality is judged relative to the set of quantumly realizable boxes 
\end{lem}

\begin{proof}
 For all boxes $B$,  (ii) $\implies$ (i) and (ii)$^{\prime}$ $\implies$ (i)$^{\prime}$ because LO is a strict subset of LOSR.  It therefore suffices to consider only the reverse implications.
Claim (a).  That (i) $\implies$ (ii) (respectively, (i)$^{\prime}$ $\implies$ (ii)$^{\prime}$) for a convexly extremal $B$ is seen as follows: if a mixture of different LO operations takes $\rho$ (respectively $B_0$) to $B$, then by the convex-extremality of $B$, every LO operation in the mixture must take $\rho$ (respectively $B_0$) to $B$.   Now consider the case where $B$ is convexly nonextremal but LO-equivalent to a convexly extremal box, which we denote by $B_{\rm ext}$.
Note that $B_{\rm ext}$ is also LOSR-equivalent to $B$, since LOSR subsumes LO.  By assumption, $\rho \mapsto B$ by LOSR (respectively $B_0 \mapsto B$ by LOSR).  It then follows from the LOSR-equivalence of 
$B$ and $B_{\rm ext}$  that $\rho \mapsto B_{\rm ext}$ by LOSR (respectively $B_0 \mapsto B_{\rm ext}$ by LOSR).  Next, from the fact that  (i) $\implies$ (ii) for convexly extremal boxes, it follows that $\rho \mapsto B_{\rm ext}$ by LO (respectively $B_0 \mapsto B_{\rm ext}$ by LO).  Finally, given the LO-equivalence of $B$ and $B_{\rm ext}$, we obtain $\rho \mapsto B$ by LO (respectively $B_0 \mapsto B$ by LO).   Claim (b).  To see that there are convexly nonextremal boxes $B$ for which  (i)$^{\prime}$ $\nimplies$ (ii)$^{\prime}$, it suffices to consider the case where $B_0$ is a product box  while $B = B_0  \otimes B_{\rm extra}$ where 
$B_{\rm extra}$ is any local box that is not a product box, so that it can be prepared for free using LOSR operations, but not using LO operations. (Here `$\otimes$' denotes parallel composition of resources \cite{coecke2016mathematical}.) To see that there are convexly nonextremal boxes $B$ for which (i) $\nimplies$ (ii), we can take  $B = B'  \otimes B_{\rm extra}$ where $B'$ is any convexly extremal box that self-tests $\rho$.
\end{proof}

Because of this lemma, some pre-existing results concerning state-to-box conversions under LO coincide with results about state-to-box conversions under LOSR.  In such cases, the LO-based assessments of which state-to-box conversions are possible coincide with the LOSR-based assessments and one can simply incorporate all previous results based on LO  into a resource theory based on LOSR. 
We will see that considerations of state-to-box conversions in discussions of the anomaly of nonlocality, studied in Sec.~\ref{sec:anomalies}, are of this sort.  Other pre-existing results, however, {\em do} need to be corrected if one is to understand them as results in a resource theory based on LOSR. 
  The precise conditions for the self-testing of states by boxes, considered in Sec.~\ref{sec:selftesting}, are of this sort.\\

\section{Resolving the anomaly of nonlocality}\label{sec:anomalies}

As summarized in the introduction and in Ref.~\cite{Methot2006anomaly}, the {\em anomaly of nonlocality} refers to the fact that there are situations wherein features of nonlocal boxes are found to be realizable by a partially entangled state but {\em not} by a maximally entangled state. 
The recognition of these anomalies~\cite{Methot2006anomaly,AnomalyExtra2005Scarani,AnomalyExtra2008Brunner,vidick2011more,AnomalyExtra2011Junge,Acin2012randomnessvsnonlocality,AnomalyExtra2014Tan,AnomalyExtra2015Augusiak,AnomalyExtra2015Fonseca,AnomalyExtra2016Bowles,AnomalyExtra2017Kabel,AnomalyExtra2018Curchod,AnomalyExtra2018Bamps,AnomalyExtra2018Chitambar,AnomalyExtra2018Lipinska,AnomalyExtra2018Barasinski} was important insofar as it made clear that it is not straightforward to understand the nonlocality of boxes and the entanglement of states as two manifestations of a single type of resource.  The aim of our article is to show that one can nonetheless do so by recasting the central questions into a  formal resource-theoretic framework.  In particular, we show that if one quantifies the entanglement properties of quantum states via LOSR operations, rather than LOCC operations, then the entanglement of states and the nonlocality of boxes are indeed seen to be two manifestations of a single type of resource, namely, nonclassicality of correlational properties.

We begin by reframing the anomaly of nonlocality within a rigorous resource-theoretic framework. 
We use the framework developed in Ref.~\cite{coecke2016mathematical}.\footnote{Recall that a set of free operations defines an ordering relation (formally, a preorder) on resources, where one resource is {\em at least as resourceful} as a second if it can be freely converted to the second. Two resources are {\em equivalently resourceful} (or in the same {\em equivalence class}) if each can be freely converted into the other, and two resources are {\em incomparable} 
 if neither is freely convertible to the other. } 
By viewing entanglement through this lens, we show that each instance of the anomaly of nonlocality can be recast as a set of claims that are not merely counterintuitive but {\em contradictory}, thereby signaling a flaw in the conceptual scheme of unformalized resource-theoretic assumptions within which they arose.\footnote{Note that earlier work on the anomaly of nonlocality did not conceive of it as a paradox that was in need of resolution.   The fact that the anomaly {\em becomes} a paradox when recast in a resource-theoretic framework helps us to identify how to achieve a unified treatment of nonlocality of boxes and entanglement of states as resources of nonclassicality of correlational properties.}

To begin, we note that it is possible to find a nonlocal box $B$ 
which can be realized
 from some partially entangled pure state of a given Schmidt rank,
\begin{align}  \label{convpartialtobox}
\ket{\psi_{\rm partial}} \mapsto B,
\end{align}
but which  {\em cannot} be realized from any maximally entangled pure state of the same Schmidt rank,\footnote{Formally, $\ket{\psi_{\rm max}}$ is any state for which the squared Schmidt coefficients
  describe a uniform distribution for the given Schmidt rank, while $\ket{\psi_{\rm partial}}$ is any state for which they describe a nonuniform distribution. 
}
\begin{align}\label{notconv}
\ket{\psi_{\rm max}} \not\mapsto B.
\end{align}

The following list provides a number of concrete examples of this phenomenon.  For each example,
we specify the box $B$ appearing therein by reference to a convex function that witnesses its nonlocality.  
For each of the following boxes, one can find a $\ket{\psi_{\rm max}}$ and a $\ket{\psi_{\rm partial}}$ of the same Schmidt rank such that Eqs.~\eqref{convpartialtobox} and \eqref{notconv} hold:
\begin{samepage}
\begin{itemize} \label{listofnonmaximaltasks}
\item a box that achieves the maximum probability
of running
Hardy's version of Bell's theorem~\cite{Hardy1993paradox}.
\item a box that maximally violates a tilted Bell inequality~\cite{Acin2002,liang2011semi,vidick2011more}, thereby offering more noise resistance for that inequality~\cite{Acin2002}. 
\item a box that has extractable secret key rate higher than $\approx 0.144$~\cite{Scarani2006QKD,Acin2006QKD}. 
\item a box that has Kullback-Leibler divergence (i.e., relative entropy distance)
from the set of local boxes larger than $\approx 0.058$~\cite{acin2005optimal}.
\end{itemize}\end{samepage}

Meanwhile, standard entanglement theory tells us that any partially entangled pure state can be realized
starting from a maximally entangled state of the same Schmidt rank~\cite{nielsen1999conditions}:
\begin{align} \label{convmaxtopartial}
\ket{\psi_{\rm max}} \mapsto \ket{\psi_{\rm partial}}.
\end{align}

It is now evident what is puzzling about these three claims (Eqs.~\eqref{convpartialtobox}, \eqref{notconv} and \eqref{convmaxtopartial}): if the conversion relations in Eqs.~\eqref{convmaxtopartial} and \eqref{convpartialtobox} hold in a resource theory, then given that resource conversion relations are necessarily transitive in any such theory\footnote{Transitivity of resource conversions is necessary in the framework of Ref.~\cite{coecke2016mathematical} because the free operations are required to be closed under sequential composition.}---i.e., if $R_1 \mapsto R_2$ and $R_2 \mapsto R_3$ then $R_1 \mapsto R_3$---it follows that we should have $\ket{\psi_{\rm max}} \mapsto B$, which contradicts Eq.~\eqref{notconv}. 

We now identify the flaw in the unformalized resource-theoretic assumptions
 that led to this contradiction. 
   It is
the implicit idea that the three conversion relations all hold relative to a single notion of resourcefulness, that is, that they all hold relative to the {\em same} set of free operations and thus can be considered as relations holding in one and the same resource theory.
In the description of the anomaly,  the claim about state-to-state conversion,
 Eq.~\eqref{convmaxtopartial}, is implicitly evaluated relative to LOCC, while the claims about state-to-box conversions, Eqs.~\eqref{convpartialtobox} and~\eqref{notconv},
   are implicitly evaluated relative to LO. 

In Section~\ref{sec:unifiedRT}, we argued that the most straightforward way of understanding the interplay of entanglement and nonlocality resource-theoretically is to imagine a network with a common source among the parties but no channels, in which case both state-to-state and state-to-box conversions should be evaluated relative to LOSR rather than LO or LOCC. As we now show, this resolves the contradiction. 
The standard claims about the state-to-box conversions are not modified when one replaces LO by LOSR.
Eq.~\eqref{convpartialtobox} holds with respect to $\LOSR$ because $\LO \subset \LOSR$, and Eq.~\eqref{notconv} holds with respect to $\LOSR$  by Lemma~\ref{convext} and the fact that, for each of the examples given, the box in question
  is convexly-extremal in the set of quantumly realizable boxes. (This follows from the fact that the functions which the boxes maximize in each example are convex-linear.)  On the other hand, the standard claim about the state-to-state conversion {\em is} modified when one replaces LOCC by LOSR.
If one judges conversion between entangled states relative to LOSR, rather than LOCC, then it is the {\em negation} of Eq.~\eqref{convmaxtopartial} that holds, namely,
\begin{align} \label{notconvmaxtopartial}
\ket{\psi_{\rm max}} \not\mapsto \ket{\psi_{\rm partial}}.
\end{align}
  This is because $\ket{\psi_{\rm max}}$ and $\ket{\psi_{\rm partial}}$  are {\em incomparable} in the resource theory of LOSR-entanglement---neither can be converted into the other under LOSR,
 as shown below (see Corollary~\ref{cor:equivalent_or_incomparable}).
 But Eq.~\eqref{notconvmaxtopartial}, unlike Eq.~\eqref{convmaxtopartial}, is {\em consistent} with Eqs.~\eqref{notconv} and \eqref{convpartialtobox}, and therefore there is no contradiction (and hence no anomaly).  
 
The terms ``partially entangled'' and ``maximally entangled'' are apt descriptions of
$\ket{\psi_{\rm partial}}$ and $\ket{\psi_{\rm max}}$ when one is considering their LOCC-entanglement properties.  
This is because $\ket{\psi_{\rm partial}}$ is {\em strictly below} $\ket{\psi_{\rm max}}$ in the LOCC order (since, in addition to Eq.~\eqref{convmaxtopartial}, we have $\ket{\psi_{\rm partial}} \not\mapsto \ket{\psi_{\rm max}}$) and consequently there exists {\em some} 
LOCC-entanglement monotone, $M_{\rm LOCC}$, for which $M_{\rm LOCC}( \psi_{\rm partial} ) < M_{\rm LOCC}( \psi_{\rm max} )$ and no LOCC-entanglement monotones relative to which this strict inequality is reversed.
When considering their LOSR-entanglement properties, however, the terminology is no longer appropriate.
   In accordance with Eq.~\eqref{notconvmaxtopartial}, there necessarily exists an LOSR-entanglement monotone, $M_{\rm LOSR}$, relative to which $M_{\rm LOSR}( \psi_{\rm partial} ) > M_{\rm LOSR}( \psi_{\rm max} )$. 
From this perspective, it is {\em natural}, rather than anomalous, that there exist tasks---such as realizing the sorts of nonlocal boxes that appear in the list presented earlier---for which the type of LOSR-entanglement required to realize the task is present in $|\psi_{\rm partial}\rangle$ but not in $|\psi_{\rm max}\rangle$.  Indeed, one can define a nontrivial
 LOSR-entanglement monotone (i.e., one that is not also an LOCC-entanglement monotone) from each example of an anomaly of nonlocality. Given a function over boxes that witnesses the type of nonlocality described in the example,  the LOSR monotone over {\em states} is simply the maximum value of that function among boxes that are LOSR-realizable starting from the given state. 
We provide the details in Appendix~\ref{monotones}.

The best known of the anomalies of nonlocality is the one concerning Hardy's version of Bell's theorem, so it is useful to reiterate our conclusion for it specifically.
The fact that the Hardy-type correlations {\em cannot} be achieved by a maximally entangled state but {\em can} be achieved by a partially entangled state surprises almost everyone who encounters the topic.  Presumably this is because---based on their familiarity with LOCC-entanglement---they expect that whatever resource of nonclassicality is present in a {\em partially} entangled state, it ought to be less than 
the resource of nonclassicality that is present in a {\em maximally} entangled state.  The resolution of the puzzle is that the notion of nonclassicality that is relevant for Bell scenarios is LOSR-entanglement, not LOCC-entanglement, and that there are measures of LOSR-entanglement relative to which what we call a partially entangled state is {\em more nonclassical} than what we call a maximally entangled state. 

The LOSR-incomparability of $\ket{\psi_{\rm max}}$ and $\ket{\psi_{\rm partial}}$ also harmonizes with the recently demonstrated~\cite{wolfe2020quantifying} LOSR-incomparability of a Tsirelson box (which provides the maximal possible quantum violation of the  Clauser-Horne-Shimony-Holt inequality~\cite{CHSH}), denoted by $B_{\rm Tsir}$, and a Hardy box (which achieves the maximum
 probability of running Hardy's version of Bell's theorem), denoted by $B_{\rm Hardy}$. 
Indeed, both instances of incomparability can be inferred directly from: (i) the transitivity of resource conversions within an LOSR resource theory incorporating both states and boxes, and (ii)
known facts about the possible and impossible state-to-box conversions under LOSR, namely, that $\ket{\psi_{\rm max}} \mapsto B_{\rm Tsir}$ while  $\ket{\psi_{\rm max}} \not\mapsto B_{\rm Hardy}$, and that $\ket{\psi_{\rm partial}} \mapsto B_{\rm Hardy}$ while $\ket{\psi_{\rm partial}} \not\mapsto B_{\rm Tsir}$.\footnote{These facts about state-to-box conversions are well-known when the operations are LO, and can be inferred to hold also for LOSR by appealing to Lemma~\ref{convext} and the convex extremality of $B_{\rm Tsir}$ and $B_{\rm Hardy}$.}
  For instance, to see that these state-to-box conversion relations imply that $\ket{\psi_{\rm max}} \not\mapsto \ket{\psi_{\rm partial}}$, it suffices to note that if it were the case that $\ket{\psi_{\rm max}} \mapsto \ket{\psi_{\rm partial}}$, then we could follow this conversion with $\ket{\psi_{\rm partial}} \mapsto B_{\rm Hardy}$ in order to have a means of converting $\ket{\psi_{\rm max}}$ to $B_{\rm Hardy}$,  thereby yielding a contradiction.  Similarly, to see that these relations imply that $B_{\rm Tsir} \not\mapsto B_{\rm Hardy}$, one merely notes that if it were the case that $B_{\rm Tsir} \mapsto B_{\rm Hardy}$, then by implementing $\ket{\psi_{\rm max}} \mapsto B_{\rm Tsir}$ followed by $B_{\rm Tsir} \mapsto B_{\rm Hardy}$, one would have a means of converting $\ket{\psi_{\rm max}}$ to $B_{\rm Hardy}$, thereby yielding a contradiction.  Similar arguments can be given to establish that $\ket{\psi_{\rm partial}} \not\mapsto \ket{\psi_{\rm max}}$ and $B_{\rm Hardy} \not\mapsto B_{\rm Tsir}$.\footnote{As an aside, this argument provides an alternative proof of the LOSR-incomparability of $B_{\rm Tsir}$ and $B_{\rm Hardy}$ to the one presented in Ref.~\cite{wolfe2020quantifying}.}

\section{Genuine multipartite entanglement}\label{sec:multipartite}

There is also some tension between results concerning genuine multipartite entanglement and those concerning genuine multipartite nonlocality when the definitions of these concepts are motivated by the LOCC paradigm for entanglement.
We again begin by reframing this tension as an outright inconsistency
by formulating a {\em genuinely multipartite anomaly of nonlocality}.  We consider the case of three parties for concreteness, although our analysis can be easily generalized  
to cases with more parties.

Denote the entangled state ${\frac{1}{\sqrt{2}}(\ket{00}+\ket{11})}$ by $\ket{\phi^{+}}$. In the preorder of tripartite entangled states relative to LOCC-convertibility, 
the state ${\ket{\psi_{\rm 2Bell}} \equiv |\phi^{+}\rangle_{A_1B}\otimes  |\phi^{+}\rangle_{A_2C}}$ is above the state ${\ket{\psi_{\rm GHZ}} \equiv \tfrac{1}{\sqrt{2}} (|000\rangle_{ABC} +|111\rangle_{ABC})}$ because the former can be deterministically converted to the latter by LOCC,
 \begin{align} \label{convBellpairtoGHZ}
\ket{\psi_{\rm 2Bell}} \mapsto \ket{\psi_{\rm GHZ}}.
\end{align}
(It suffices for one party to prepare three systems in their lab in the state $\ket{\psi_{\rm GHZ}}$, and then to use $\ket{\psi_{\rm 2Bell}}$ to teleport one part to each of the other two parties, which requires classical communication.)
And yet, 
 there are tripartite boxes, such as the Mermin box~\cite{Mermin1990,Brassard2005},
that can be realized from $\ket{\psi_{\rm GHZ}}$ by local measurements,
\begin{align} \label{convGHZtomaxMermin}  
\ket{\psi_{\rm GHZ}} \mapsto B_{\rm Mermin},
\end{align}
but that cannot be so realized from $\ket{\psi_{\rm 2Bell}}$,
\begin{align} \label{notconvBellpairtomaxMermin}
\ket{\psi_{\rm 2Bell}} \not\mapsto B_{\rm Mermin},
\end{align}
as follows from a result in Ref.~\cite{wolfe2021infl}.\footnote{The relevant result from Ref.~\cite{wolfe2021infl} is that the Mermin box cannot be achieved by LOSR processing from
{\em any} state, $\rho_{\rm triangle}$, that is a tensor product of states having entanglement between {\em pairs} of parties only (i.e., generated from entangled sources consistent with the so-called ``triangle scenario'' network), $\rho_{\rm triangle} \not\mapsto B_{\rm Mermin}$. Eq.~\eqref{notconvBellpairtomaxMermin} follows because $\ket{\psi_{\rm 2Bell}}$ is an instance of such a state.}
As before, Eqs.~\eqref{convBellpairtoGHZ}, \eqref{convGHZtomaxMermin}, and ~\eqref{notconvBellpairtomaxMermin} seem to imply a contradiction 
 given the transitivity of resource conversion relations.

The resolution of the puzzle proceeds as in the case of the bipartite anomalies.  The reason for the seeming contradiction is that the conversion relations of Eqs.~\eqref{convGHZtomaxMermin} and \eqref{notconvBellpairtomaxMermin} are implicitly evaluated relative to LO, 
while that of Eq.~\eqref{convBellpairtoGHZ} is evaluated relative to LOCC.  If, however, one evaluates 
all conversion relations
relative to LOSR,
 then although Eqs.~\eqref{convGHZtomaxMermin} and \eqref{notconvBellpairtomaxMermin} do not change, by virtue of Lemma~\ref{convext} and the fact that the Mermin box is convexly extremal,
Eq.~\eqref{convBellpairtoGHZ} does.  Specifically, relative to LOSR, 
  $\ket{\psi_{\rm 2Bell}}$ and $\ket{\psi_{\rm GHZ}}$ are found to be {\em incomparable}, so that the negation of Eq.~\eqref{convBellpairtoGHZ} holds, that is,
\beq\label{noconvBellGHZ}
\ket{\psi_{\rm 2Bell}} \not\mapsto \ket{\psi_{\rm GHZ}},
\eeq
and the contradiction is blocked. The proof of incomparability follows from a condition for LOSR-convertibility we derive further on (see Corollary \ref{cor:multipartite_extension}),
as is made explicit in Appendix~\ref{BellGHZ}.

This anomaly and its resolution shed light on how one ought to define notions of genuine multipartite entanglement and nonlocality, and specifically what it means for these to be {\em genuine 3-way}.

We begin by presenting a slightly different perspective on the anomaly.
Consider the conventional definition of {\em genuine 3-way entanglement}~\cite{kProducibleEntanglement,EntanglementManyBody}.  Recalling that a tripartite state is termed {\em biseparable} if there is a partitioning of the parties into two groups such that the state is separable relative to this bipartition, the standard definition can
be simply stated as follows: 
\begin{defn}[standard]
A tripartite state is {\em genuine 3-way entangled} if and only if it is not a mixture of biseparable states.
\end{defn}

Now note that $\ket{\psi_{\rm 2Bell}}$ counts as genuine 3-way entangled by this definition.  This is somewhat counterintuitive a priori, given that $\ket{\psi_{\rm 2Bell}}$ is composed of states that contain only {\em bipartite} entanglement.  One can glean some insight into what led to the adoption of the conventional definition, in spite of its counterintutive features, from the fact that $\ket{\psi_{\rm 2Bell}}$ is above $\ket{\psi_{\rm GHZ}}$ in the LOCC order (Eq.~\eqref{convBellpairtoGHZ}).  Because it is generally agreed that any reasonable
 definition of genuine 3-way entanglement should be such that $\ket{\psi_{\rm GHZ}}$ counts as genuine 3-way entangled, it follows that because $\ket{\psi_{\rm 2Bell}}$ is above $\ket{\psi_{\rm GHZ}}$ in the order, it  too must qualify as having such entanglement.

With the conventional notion of genuine 3-way entanglement in mind,
we are now in a position to present the alternative perspective on the genuinely tripartite anomaly of nonlocality.
First, note that it is generally agreed that any definition of genuine 3-way nonlocality should be such that $B_{\rm Mermin}$  counts as genuine 3-way nonlocal.
But because
 $\ket{\psi_{\rm 2Bell}}$ is {\em above} $\ket{\psi_{\rm GHZ}}$ in the LOCC order, one would expect that whatever resource of genuine 3-way entanglement is required to generate the genuine 3-way nonlocality inherent in $B_{\rm Mermin}$, $\ket{\psi_{\rm 2Bell}}$ would have it if $\ket{\psi_{\rm GHZ}}$ does.  And yet this intuitive conclusion is in conflict with
 Eq.~\eqref{notconvBellpairtomaxMermin}.

The problem is that the conventional notion of genuine 3-way entanglement is motivated by LOCC.
 We therefore propose an alternative definition of genuine 3-way entanglement, motivated by LOSR.  We also show that there is a corresponding alternative definition for genuine 3-way nonlocality which mirrors our alternative definition of genuine 3-way entanglement.  The genuinely tripartite anomaly is shown to admit of a natural resolution relative to these two notions (genuine 3-way entanglement and genuine 3-way nonlocality), demonstrating that these have a more natural interplay than exists between the conventional pair of notions.

In the LOSR paradigm, the distinction between a classical and a quantum resource shared among some parties is the distinction between sharing classical randomness and sharing entanglement.  
Consequently, the natural manner of defining a resource of 3-way nonclassicality is to consider resourcefulness relative to a set of operations wherein all 2-way nonclassicality is considered free---that is, where 2-way common sources are allowed to be quantum---while the 3-way common source is required to be classical. 
Thus, we define genuine 3-way {\em nonclassicality} for quantum states and for nonlocal boxes (as well as for other types of multipartite processes, such as quantum measurements, quantum channels,  and multi-time processes) as those that are nonfree relative to 
  {\em LOSR supplemented by 2-way shared entanglement (LOSR2WSE)}.\footnote{One can define  genuinely $k$-way nonclassicality among $n$ parties (where $k\leq n$) in an analogous manner: via LOSR supplemented by $(k-1)$-way shared entanglement~\cite{navascues2020genuine}.}     
  \begin{defn}[resource-theoretic]\label{Defn:G3WE}
A nonlocal box is {\em genuine 3-way nonlocal} and an entangled state is {\em genuine 3-way entangled} if and only if they cannot be obtained by LOSR2WSE, that is, local operations together with a source of correlations consisting of shared entanglement between each pair of parties and shared randomness among all three. 
    \end{defn}
   These notions
    are illustrated in Figs.~\ref{genuinely3way} and ~\ref{non_locality}.
Ref.~\cite{wolfe2021infl}  discusses boxes that are genuine 3-way nonlocal in this sense, and Ref.~\cite{navascues2020genuine} discusses states that are genuine 3-way entangled in this sense. The causal structure 
of the free resources
 associated to LOSR2WSE
 has been termed the ``quantum triangle scenario with shared randomness''~\cite{wolfe2021infl,kraft2021quantum}.

\begin{figure}[htb!]
\centering
    \includegraphics[scale=1]{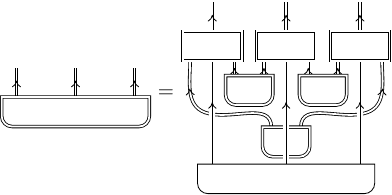}
    \caption{
       A tripartite entangled state is genuine 3-way entangled if it cannot be decomposed as shown; that is, if it cannot be realized using LOSR supplemented by 2-way shared entanglement (LOSR2WSE). 
        Throughout, double wires represent quantum systems and single wires represent classical systems. 
    }
    \label{genuinely3way}
\end{figure}

\begin{figure}[htb!]
\centering
    \includegraphics[scale=1]{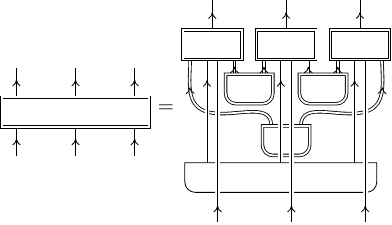}
    \caption{
           A tripartite box is genuine 3-way nonlocal if it cannot be decomposed as shown; that is, if it cannot be realized using LOSR supplemented by 2-way shared entanglement (LOSR2WSE).  }
    \label{non_locality}
\end{figure}

 It is clear that $\ket{\psi_{\rm 2Bell}}$ can be realized via LOSR2WSE,
  while it can be shown that $\ket{\psi_{\rm GHZ}}$ cannot (see  Ref.~\cite{navascues2020genuine}), and so of the two, only $\ket{\psi_{\rm GHZ}}$ is genuine 3-way entangled according to Definition~\ref{Defn:G3WE}.
   In this approach, therefore, the intuitions we noted earlier regarding genuine 3-way entanglement are vindicated: if a state can be obtained from one where all of the entanglement is of the bipartite variety, then it is not genuine 3-way entangled.  

 It is also the case that 
 $B_{\rm Mermin}$ cannot be realized by LOSR2WSE, as follows from results
 in Ref.~\cite{wolfe2021infl}, so that  $B_{\rm Mermin}$ is genuine 3-way nonlocal according to Definition~\ref{Defn:G3WE}
   (just as it was relative to the old definition).
 
In the description of the tripartite anomaly, the fact that $\ket{\psi_{\rm GHZ}}$ can be converted into $B_{\rm Mermin}$ (Eq.~\eqref{convGHZtomaxMermin})  while $\ket{\psi_{\rm 2Bell}}$ {\em cannot} (Eq.~\eqref{notconvBellpairtomaxMermin})
 was only surprising relative to the belief that $\ket{\psi_{\rm 2Bell}}$ must have {\em more} genuine 3-way nonclassicality than $\ket{\psi_{\rm GHZ}}$ does, on the grounds that it is above $\ket{\psi_{\rm GHZ}}$ in the LOCC order.\footnote{This is in precise analogy to how the fact that $\ket{\psi_{\rm partial}}$ can be converted to a box manifesting Hardy-type correlations (Eq.~\eqref{convpartialtobox}) while $\ket{\psi_{\rm max}}$ cannot (Eq.~\eqref{notconv}) is only surprising relative to the belief that $\ket{\psi_{\rm max}}$ has more nonclassicality than $\ket{\psi_{\rm partial}}$ does.} But this is overturned in the approach just described, since  $\ket{\psi_{\rm 2Bell}}$ explicitly has {\em less} genuine 3-way nonclassicality than $\ket{\psi_{\rm GHZ}}$ does, since $\ket{\psi_{\rm 2Bell}}$ has none while $\ket{\psi_{\rm GHZ}}$ has some.
 Given that $B_{\rm Mermin}$ has genuine 3-way nonclassicality according to our definition, it becomes intuitively clear why  $\ket{\psi_{\rm GHZ}}$ and not $\ket{\psi_{\rm 2Bell}}$ can be converted into $B_{\rm Mermin}$.

Furthermore,
 we note that whereas our definition of genuine 3-way nonclassicality 
fits within the mathematical framework for resource theories~\cite{coecke2016mathematical}---as the property of being nonfree relative to LOSR2WSE---the conventional definition does not.  
The latter fact is most easily seen by recalling an awkward feature of the conventional definition,
namely, that the set of states that are {\em not} genuinely multipartite entangled is not closed under tensor products.  For example, consider the tripartite states $|\phi^{+}\rangle_{A_1B} \otimes |0\rangle_C$ and $|\phi^{+}\rangle_{A_2 C} \otimes |0\rangle_B$.  Both are biseparable (relative to different partitions of the tripartite system).  
Jointly having these states, however, is 
equivalent to having $|\psi_{\rm 2Bell}\rangle = |\phi^{+}\rangle_{A_1B} \otimes |\phi^{+}\rangle_{A_2 C}$, which, as noted previously, is not biseparable 
   and therefore {\em is} genuine 3-way entangled according to the conventional definition.  

Although this feature of the conventional definition has been acknowledged by some as counterintuitive and somewhat perverse~\cite{navascues2020genuine,contreras2021,MultipartiteProblem2020Luo}, we wish to draw attention here to the fact that it is {\em inconsistent} with the framework of resource theories\footnote{Just as the conversion relations in each of the anomalies were inconsistent with the transitivity of resource conversion relations.}, because the latter 
 stipulates that the set of free resources must be closed under parallel composition (see Definition 2.1 of Ref.~\cite{coecke2016mathematical}). From the resource-theoretic perspective, therefore, the fact that the property of biseparability is not preserved under parallel composition
  implies that biseparability is simply not a viable candidate for the property that defines the set of free resources in a resource theory. Hence, 
 nonbiseparability is not a viable candidate for the property that defines 
 the resource of genuinely multipartite nonclassicality.
 
Although we have here been primarily concerned with making the case that it is the notion of genuine 3-way entanglement based on  LOSR2WSE, rather than the one based on LOCC, that does best justice to our intuitions about multipartite Bell scenarios,
 we have also seen that the conventional notion fails to satisfy certain desiderata that one would want any resource to satisfy.  This raises the question of whether it even makes sense,  within the LOCC paradigm, to speak of a resource of entanglement that is genuine 3-way.  The following considerations suggest that it does not.

We begin by highlighting why such a notion {\em does} make sense within the LOSR paradigm of entanglement.  There, 
the resource of entanglement constitutes quantumness of common sources.
Consequently, the fact that there is a distinction between quantum common sources between each pair of parties
 and a quantum common source between 
  all three parties implies that there is a distinction between 2-way and 3-way notions of entanglement.  In particular, if all 2-way quantum common sources
   are taken to be freely available, there is still something that the parties can lack, namely, 3-way quantum common sources.

 In the LOCC paradigm, by contrast, the resource of entanglement is equivalent to a resource of quantum communication 
  and there is simply nothing corresponding to the notion of quantum communication between all the parties beyond the notion of quantum communication between any pair of parties.  In other words, if all 2-way quantum communication channels are taken to be freely available, there is no communication resource that the triple of parties lack.  Such considerations suggest that it may not be sensible to try and define a notion of genuine 3-way entanglement in the LOCC paradigm.

\section{Self-testing of entangled states}\label{sec:selftesting}

In this section, we show that LOSR-entanglement is the notion that best captures many existing ideas regarding self-testing,
  i.e., the certification of the presence of particular entangled states by observing a nonlocal box~\cite{mayers1998quantum,mayers2003self,vsupic2020self,Scarani2019}.
  \footnote{
  Note that it is common to refer to a box as self-testing {\em both} a state and measurements, rather than a state alone~\cite{vsupic2020self,JedWeak2020,ScaraniStateSelfTest2009}.  We will discuss this distinction further in Appendix~\ref{Oversight}. Note also that we are here considering only the notion of {\em perfect} self-testing, rather than the notion of {\em robust} self-testing~\cite{McKague2012,Yang2013selftesting,vsupic2020self,Scarani2019}, which we do not discuss in this article.  Similarly, we do not consider self-testing based on observing a maximal violation of a Bell inequality, rather than based on observing a particular box~\cite{ScaraniStateSelfTest2009,Bamps2015selftesting,Baccari2020}.} 
  More precisely, we show that many results on self-testing can be rederived in a methodologically sound manner within the resource theory based on LOSR, and that simple generalizations of these results (including some correctives) follow directly.

Colloquially, a state is said to be self-tested by a box if this state is ``the only one'' from which the box can be obtained by implementing local measurements. However, there is {\em never} just a single state that can yield a given box, and so the standard definition allows for some freedom in the specification of the state.  The freedom that is implied by the original definition~\cite{mayers1998quantum,mayers2003self,Yang2013selftesting,Wang2016} is this:
\begin{defn}[standard]\label{defnselftesteOld}
The pure state $|\psi\rangle$ is self-tested by a  box $B$  if $B$ can be obtained from $|\psi\rangle$ by local measurements and if for any state $\sigma$ from which $B$ can be obtained by local measurements, there is a local isometry that takes $\sigma$ to $|\psi\rangle\langle \psi| \otimes \omega$ for some state $\omega$.
A state $|\psi\rangle$ is said to be {\em self-testable} when there exists a box $B$ such that $|\psi\rangle$ is self-tested by $B$.
 \end{defn}

In Appendix~\ref{Oversight}, we discuss a subtlety regarding the claim that this is the `standard' definition of self-testing of states. Specifically, we note that various authors~\cite{Coladangelo2017,vsupic2020self} did not explicitly require 
 the condition that $B$ can be obtained from $|\psi\rangle$ by local measurements, and we explain why we deem it likely that 
these authors were assuming it implicitly.

As we did with the anomalies of nonlocality, we will begin by taking a resource-theoretic perspective on self-testing, where the resource theory is based on LOSR.  We take the natural definition of self-testing to be this: 
\begin{defn}[resource-theoretic, LOSR-based]\label{defnselfteste}
{
\abovedisplayskip=1pt plus 3pt
\belowdisplayskip=1pt plus 3pt
\abovedisplayshortskip=1pt plus 3pt
\belowdisplayshortskip=1pt plus 3pt
Given a pair consisting of a density operator $\rho$ and a box $B$, we say that $\rho$  {\em is self-tested by} $B$ if it holds that
\begin{align}
\rho \mapsto B\nonumber
\end{align}
 and that for all density operators $\sigma$,
\begin{align} 
\textrm{if }\; \sigma \mapsto B\; \textrm{ then }\; \sigma \mapsto \rho,\nonumber
\end{align}
where all conversions are evaluated relative to LOSR. A state $\rho$ is said to be {\em self-testable} if there exists some box $B$ that self-tests $\rho$.
}\end{defn}

In Appendix~\ref{equivdefnself}, we give an equivalent version of this definition in terms of the notion of the \emph{upward closure} of a resource, that is, the set of all resources that are above the given resource in the preorder. Specifically, $\rho$ is self-tested by $B$ if the upward closure of $B$ among states contains only those states which are also in the upward closure of $\rho$. 
Our definition of self-testing also generalizes immediately to other sorts of objects and even to other resource theories, as we discuss in Appendix~\ref{equivdefnself}.

Note that all notions of conversion appearing in standard discussions of self-testing (e.g., the notion of `can be obtained from' in Definition~\ref{defnselftesteOld})
 are understood, within a resource-theoretic approach, as conversion {\em relative to the free operations in the resource theory.}
If self-testing is to be viewed as the task of certifying the {\em nonclassicality of correlations} in a quantum state by the observation of the nonclassicality of correlations in a box,
 then it follows from the arguments in Section~\ref{sec:unifiedRT} that the appropriate set of free operations is given by LOSR (rather than, e.g., local isometries or local operations).

The resource-theoretic perspective on self-testing makes it clear that while there is no way to be certain which equivalence class of states was the one that generated the box (because there is an unbounded number of distinct such classes in the upward closure of any box), nonetheless, one {\em can} be certain about the identity of the {\em least} of these classes (i.e., the classes in this upward closure that are lowest in the partial order). 
 As such, while much of the language used in the self-testing literature suggests that when a state is self-tested, it is being uniquely certified 
  {\em up to equivalence}, in fact it is only being uniquely certified 
  {\em up to upward closure}.\footnote{Note, however, that if one is willing to introduce a restriction on the set of states under consideration, then one {\em can} achieve certification of the state up to equivalence.  We discuss this type of certification in Appendix~\ref{selftestinguptoequivalence}.}

We now discuss the relation between our definition of self-testing (Definition~\ref{defnselfteste}) and the standard one (Definition~\ref{defnselftesteOld}).

To do so, it is useful to begin by describing how the standard definition is usually generalized to the case where the state to be self-tested is mixed, because only then can one understand why it is generally asserted that mixed states {\em cannot} be self-tested.
 This is done simply by replacing $|\psi\rangle$ in Definition~\ref{defnselftesteOld}  with a mixed state $\rho$.  In order to highlight the differences with 
 Definition~\ref{defnselfteste}, we also recast the mixed-state generalization of the standard definition into a more resource-theoretic form: 
\begin{defn}[standard, generalized to mixed states]
\label{defnselftesteMixed}
Of the form of Definition~\ref{defnselftesteOld}, but where $|\psi\rangle$ is replaced by a mixed state $\rho$.
Equivalently, of the form of Definition~\ref{defnselfteste}, but where the conversions $\rho \mapsto B$ and $\sigma \mapsto B$ are evaluated relative to LO, while the conversion $\sigma \mapsto \rho$ is evaluated relative to local isometries with the constraint that the final state factorizes with respect to the division between the system on which $\rho$ is defined and the auxiliary system (i.e., there exists a state $\omega$ 
 such that $\sigma \mapsto \rho \otimes \omega$ by a local isometry).
\end{defn}

To see that the claimed resource-theoretic recasting is accurate, note that the means by which $B$ is stipulated to be obtainable from $\sigma$ or $\rho$ in Definition~\ref{defnselftesteOld} (after substituting $\rho$ for $|\psi\rangle$)
  is through `local measurements'.
  But given that a local measurement is just  an instance of a local operation, namely, the type of local operation that takes a state a box, it follows that the state-to-box conversion relations $\rho \mapsto B$ and $\sigma \mapsto B$ 
are being evaluated relative to LO.

By contrast, after substituting a mixed state $\rho$ for $|\psi\rangle$ in Definition~\ref{defnselftesteOld},  the conversion $\sigma \mapsto \rho$  cannot be understood as being evaluated relative to LO. 
To see this, we appeal to the following lemma, which will also serve to clarify where the standard definition and its usual mixed-state generalization  coincide with an LO-based definition and where they  differ from it.

 \begin{samepage}
\begin{lem}
    \label{lem:conversion2}
    Consider the following statements about interconversion between $n$-partite states $\sigma$ and $\rho$:
    \begin{compactenum}[(i)]
        \item  $\sigma \mapsto \rho $ by LO,\\
         or equivalently,\\
          $\exists W : \rho = {\rm Tr}_{\rm aux}( W)$ and $\sigma \leftrightarrows
         W $ by local isometries.
        \item  
        $\exists \omega : \sigma \leftrightarrows
          \rho  \otimes \omega$ by local isometries. 
    \end{compactenum}
 (Here, $\leftrightarrows$ denotes convertibility in both directions, $\omega$ is a state on an auxiliary system, $W$ is a joint state on the auxiliary system and the system associated to $\rho$, and ${\rm Tr}_{\rm aux}$ denotes a partial trace over the auxiliary system.). The following implications hold among these conditions:
 \begin{enumerate}
 \item[(a)] If $\rho$ is pure, then (i) and (ii) are equivalent.
\item[(b)] If $\rho$ is mixed, then although it is still the case that (ii) $\implies$ (i), it can happen that (i) $\nimplies$ (ii).
 \end{enumerate} 
\end{lem}
\end{samepage}
 Note that the second form of condition (i) follows from considering the Stinespring dilations of the local operations.   
 
\begin{proof}
  (ii) $\implies$ (i) for all states because if there exists an $\omega$ satisfying condition (ii), then there exists a $W$ satisfying condition (i), namely, $W= \rho \otimes \omega$.  It therefore suffices to consider the implication (i) $\implies$ (ii).  Claim (a).  (i) $\implies$ (ii) in the case where $\rho$ is pure because in this case, the only joint states $W$ having $\rho$ as their reduction are of the form ${W=\rho \otimes \omega}$ for some $\omega$ (i.e., the purification of a pure state necessarily factorizes).  Claim (b).  To see that it can happen that (i) $\nimplies$ (ii) when $\rho$ is mixed, it suffices to note that there are joint states $W$ having $\rho$ as their reduction that are {\em not} of the form ${W=\rho \otimes \omega}$ for some $\omega$.  That is, if $\rho$ is mixed, then it can arise as the reduction of a state $W$ exhibiting correlations between the auxiliary system and the system on which $\rho$ is defined.
\end{proof}

Note that {\em if} the proposal for the mixed-state generalization of the standard definition (Definition~\ref{defnselftesteMixed})  had been that $\sigma$ must map
 by a local isometry to a state $W$ where $\rho = {\rm Tr}_{\rm aux}(W)$ {\em without} requiring that $W$ be of the factorizing form $\rho \otimes \omega$, then  the conversion $\sigma \mapsto \rho$ would simply be evaluated relative to LO (simply recall the two forms of condition (i) in Lemma~\ref{lem:conversion2}).  This alternative proposal for the mixed-state generalization of the standard definition, when  
 recast resource-theoretically, is simply this:
\begin{defn}[resource-theoretic, LO-based]
\label{defnselftesteLObased}
Of the form of Definition~\ref{defnselfteste}, but where all conversions are evaluated relative to LO. 
\end{defn}

In our view, this alternative approach to adapting the standard definition to mixed states is the more natural of the two, 
because what is important in self-testing is that one is able to recover the state $\rho$ from $\sigma$, while the question of whether or not the state $W$ (from which $\rho$ is obtained by partial trace) stipulates that there is correlation between the system of interest and the auxiliary system seems irrelevant. 

It turns out that there are questions for which it is significant whether a proponent of the standard definition adopts Definition~\ref{defnselftesteMixed} or Definition~\ref{defnselftesteLObased} as its mixed-state generalization. Specifically, whereas it is commonly claimed that {\em it is impossible to self-test mixed states}~\cite{Coladangelo2017,Supi2018Mutliselftest,vsupic2020self}, this claim is justified {\em only} relative to 
Definition~\ref{defnselftesteMixed}, but not 
Definition~\ref{defnselftesteLObased}.

To see this, it suffices to consider the proof that is standardly given for this claim 
 (see, e.g., Ref.~\cite{vsupic2020self}).  It is first noted that for any box $B$ that can be obtained from a mixed state $\rho$, there is a pure state $|\psi\rangle$ of the same dimension as $\rho$ that can also be used to obtain $B$ (as shown in Ref.~\cite{Sikora2016HSD}).  Therefore, among the $\sigma$ from which $B$ can be obtained by local measurements, there are some which are pure.  But for any pure $\sigma$, it is impossible to satisfy the condition that there exists an $\omega$ such that $\sigma \mapsto \rho \otimes \omega$ by a local isometry unless $\rho$ is also pure. If $\rho$ is not pure, then it follows that this condition cannot be satisfied; hence, mixed states cannot be self-tested.
 \footnote{It is, in fact, the structure of this argument which reveals that it is Definition~\ref{defnselftesteMixed} rather than Definition~\ref{defnselftesteLObased} that is presumed to be the mixed-state generalization of the standard definition of self-testing in works such as Ref.~\cite{vsupic2020self}.}

This proof does not go through if one allows the final state $W$ to exhibit correlation between the auxiliary system and the system on which $\rho$ is defined because although the purity of $\sigma$ implies that $W$ needs to be pure, a pure $W$ is compatible with the reduced state $\rho$ being mixed ($W$ simply needs to be entangled).

Indeed, if one adopts the alternative mixed-state generalization of the standard definition wherein the conversion $\sigma \mapsto \rho$ is evaluated relative to LO (Definition~\ref{defnselftesteLObased}), then one finds that mixed states {\em can} be self-tested.  It suffices only to note that there are mixed states in the LO-equivalence class of a given pure state.  Appendix~\ref{discrepanciesLOSRLO} provides an example. 
Note that the possibility of self-testing mixed states holds also under the LOSR-based definition of self-testing (Definition~\ref{defnselfteste}) given that LOSR subsumes LO, so that the LOSR-equivalence classes include the LO-equivalence classes.

It is similarly clear that 
convexly nonextremal boxes can serve to self-test states because there are convexly nonextremal boxes in the equivalence class of a given convexly extremal box (whether equivalence is judged relative to LO or LOSR).  Appendix~\ref{discrepanciesLOSRLO} provides an example. \footnote{This observation does not seem to conflict with any prior claims regarding the self-testing of states.  Although it has previously been claimed that convexly nonextremal boxes do {\em not} provide an opportunity for self-testing~\cite[Appendix C]{geometry2018} in the context of joint self-testing of state-measurement pairs (see Definition~\ref{defnselftesteOriginal} of Appendix~\ref{Oversight} in this manuscript), this claim does not conflict with the fact that convexly nonextremal boxes {\em do} provide an opportunity for self-testing of states alone.  Note that we have not attempted to see what a resource-theoretic perspective implies for how to define self-testing of state-measurement pairs, although we suspect that if one were to do so, one would find that the standard definition (Definition~\ref{defnselftesteOriginal} of Appendix~\ref{Oversight}) is also in need of revision.}

We have shown how one can extend the scope of self-testing to mixed states and to convexly nonextremal boxes. 
What accounts for the fact that such extensions  were not previously recognized?  The answer, we believe, is that without the resource-theoretic perspective on self-testing, the deficiencies in the standard definition were difficult to see.  In our view, therefore, the extension of scope of self-testing is an application of adopting the resource-theoretic perspective on self-testing.

Nonetheless, even if one grants that self-testing should be defined resource-theoretically, it is easy to see that 
 the closest resource-theoretic counterpart to the standard definition is Definition~\ref{defnselftesteLObased}, corresponding to taking LO as the set of free operations, not LOSR.  We therefore turn our attention now to articulating the precise differences that exist between these two choices, and to explaining why it is the LOSR-based definition that should be endorsed.

  As we noted in Section~\ref{sec:unifiedRT}, the tendency to {\em not} consider supplementing LO with shared randomness is widespread in the literature on Bell's theorem, and is likely 
a historical artifact of the conventional way of conceptualizing a Bell experiment.  
But as we argued in Section~\ref{sec:unifiedRT}, given that everyone agrees that entanglement and nonlocality are supposed to capture the {\em nonclassicality} of the correlational properties of states and boxes respectively, one has no choice but to include classical correlations---that is, shared randomness---in the set of free operations.   We conclude, therefore, that if there are any discrepancies in the verdicts regarding self-testing when considered relative to LO and when considered relative to LOSR, it is the verdicts based on LOSR that should be considered the correct ones.  In the following, we will show how such discrepancies can indeed arise.

The following lemma, which is the analogue of Lemma~\ref{convext} for state-to-state conversions, is useful here:
 \begin{samepage}
\begin{lem}
    \label{lem:conversion3}
    Consider the following statements about interconversion between $n$-partite states $\sigma$ and $\rho$
    \begin{compactenum}[(i)]
        \item  $\sigma \mapsto \rho $ by LOSR,
        \item  $\sigma \mapsto \rho $ by LO.
    \end{compactenum}
The following implications hold among these conditions:
 \begin{enumerate}
 \item[(a)] If $\rho$ is pure or $\rho$ is a mixed state that is LO-equivalent to a pure state, then  (i) and (ii) are equivalent.
\item[(b)] If $\rho$ is an arbitrary mixed state, then although (ii) $\implies$ (i), it can happen that (i) $\nimplies$ (ii).
 \end{enumerate} 
\end{lem}
\end{samepage}

\begin{proof}  (ii) $\implies$ (i) for all $\rho$ because every LO operation is an instance of an LOSR operation. It suffices, therefore, to consider the implication (i) $\implies$ (ii).  Claim (a).  The proof that (i) $\implies$ (ii) in the case where $\rho$ is pure is a generalization of its proof when one furthermore assumes that $\sigma$ is pure, which is  presented as Lemma~\ref{lem:conversion} of Section~\ref{sec:RTLOSR} and proven there.  As the generalization from pure $\sigma$ to mixed $\sigma$ is straightforward, we leave it to the reader to verify.  Next, suppose that $\rho$ is a mixed state that is LO-equivalent to a pure state, denoted $|\psi\rangle\langle \psi|$.  Note that $|\psi\rangle\langle \psi|$ is then also {\em LOSR-equivalent} to $\rho$, since LOSR subsumes LO.  By assumption, $\sigma \mapsto \rho$ by LOSR, so it follows from the LOSR-equivalence of $\rho$ and $|\psi\rangle\langle \psi|$  that $\sigma \mapsto |\psi\rangle\langle \psi|$ by LOSR.  Next, from the fact that  (i) $\implies$ (ii) for pure states, it follows that $\sigma \mapsto |\psi\rangle\langle \psi|$ by LO.  Finally, given the LO-equivalence of $\rho$ and $|\psi\rangle\langle \psi|$, we obtain $\sigma \mapsto \rho$ by LO. Claim (b). To see that there are mixed states $\rho$ for which (i) $\nimplies$ (ii), it suffices to consider the case where $\sigma$ is a product state, while $\rho=\sigma \otimes \zeta$ where $\zeta$ is any separable state that is not a product state (so that it can be prepared for free using LOSR operations, but not using LO operations). 
\end{proof}

This lemma implies the existence of cases in which there is a discrepancy between verdicts about self-testing using the LO-based and the LOSR-based approaches, and stipulates the conditions under which there is no such discrepancy.  The result is as follows:
\begin{prop}\label{prop:agree}
Consider the question of whether a box $B$ self-tests a state $\rho$ or not. 
\begin{enumerate}
\item[(a)] If $\rho$ is a pure state or a mixed state that is LO-equivalent to a pure state and $B$ is a convexly extremal box or a convexly nonextremal box that is  
LO-equivalent to a convexly extremal box, then the answer implied by the  LOSR-based definition of self-testing (Definition~\ref{defnselfteste}) coincides with the answer implied  by the LO-based definition (Definition~\ref{defnselftesteLObased}).  
\item[(b)] However, there are examples of state-box pairs wherein the state $\rho$ is mixed or the box $B$ is convexly nonextremal  such that the answers implied by the LOSR-based and LO-based definitions diverge.
\end{enumerate}
\end{prop}

\begin{proof}  Claim (a).  This follows  from the relevant definitions of self-testing together with claim (a) of Lemma~\ref{convext} and claim (a) of Lemma~\ref{lem:conversion3}.  
Claim (b). To find an example of a mixed state $\rho$ that is self-tested by a box $B$ relative to LOSR but not relative to LO, it suffices to note that given a pure state $|\psi\rangle$ that is self-tested by $B$, one can construct a mixed state $\rho$ that is LOSR-equivalent to $|\psi\rangle$ but not LO-equivalent to it.  An example is provided in Appendix~\ref{discrepanciesLOSRLO}. Similarly, to find an example of a box $B$ that self-tests a state $\rho$  relative to LOSR but not relative to LO, it suffices to note that given a convexly extremal box $B_{\rm ext}$ that self-tests $\rho$, one can construct a convexly nonextremal box $B$ that is LOSR-equivalent to $B_{\rm ext}$ but not LO-equivalent to it.   An example is provided in Appendix~\ref{discrepanciesLOSRLO}.
\end{proof}

Note that because the resource-theoretic definition of self-testing can be framed entirely in terms of {\em equivalence classes} of states and boxes rather than in terms of the states and boxes themselves (see Definition~\ref{AbstractDefnSelfTesting2} in Appendix~\ref{equivdefnself}), the fact that every pure state has mixed states in its LOSR-equivalence class and that every convexly extremal box has convexly nonextremal boxes in its LOSR-equivalence class implies that the convex extremality of resources is not necessary for them to appear in nontrivial self-testing relations. 
The notion of extremality that is most relevant for self-testing is instead one that concerns the position of a resource in the preorder over resources.   

Thus, our resource-theoretic approach to self-testing not only shows that mixed states and convexly nonextremal boxes can be self-tested, it makes clear that the scope of these includes not only the states and boxes that are LO-equivalent to their pure counterparts, but also those that are LOSR-equivalent to these.

Note that for the case of pure states and convexly extremal boxes, claim (a) of Lemma~\ref{lem:conversion2} implies that the LO-based definition of self-testing (Definition~\ref{defnselftesteLObased}) is equivalent to the standard definition (Definition~\ref{defnselftesteOld}).  Combined with Proposition~\ref{prop:agree}, this implies that the LOSR-based definition of self-testing (Definition~\ref{defnselfteste}) is also equivalent to the standard definition in this case.  Even in this case, however, the resource-theoretic approach yields some novel insights relative to the standard approach.

Specifically, while the standard definition might at first glance seem {\em ad hoc}, particularly with its appeal to a local isometric freedom, the resource-theoretic approach can provide a justification of its form.
It is seen to follow from two plausible claims underlying the resource-theoretic definition: (i) that 
  the essence of self-testing a resource $R$ by a resource $R'$ is certifying that $R$ is the unique least resource (up to equivalence) in the upward closure of $R'$
(as specified formally in Definition~\ref{AbstractDefnSelfTesting2} of Appendix~\ref{equivdefnself}),
  and (ii)    that for the particular case of self-testing entanglement by nonlocality,
the free operations relative to which convertibility relations should be judged are LOSR operations.

The utility of having a resource-theoretic
 justification of the definition of self-testing of states
is highlighted by the fact that it provides clear answers to challenges regarding self-testing in multipartite Bell scenarios. Consider the example of chiral states.
Chiral pure states---states which are inequivalent under local unitary transformations to their complex conjugate---are often proposed as examples of states that cannot be self-tested, since a chiral state and its complex conjugate give rise to the same set of boxes.\footnote{A similar situation arises for states which are not LOSR-equivalent to their partial transposes; see the results of Ref.~\cite{Hiraxh2020PPT}. Note also that chiral pure states exist even for three-qubit systems~\cite{Acin2000triclassify,Acin2001puretriclassify}.} A similar failure of self-testability arises for {\em measurements} rather than states~\cite{mckague2010generalized}.
However, some authors have proposed that the definition of self-testing be modified to include complex conjugation in addition to the  local isometric freedom, so that such states and measurements could also be said to be self-tested. In the context of self-testing of measurements, Ref.~\cite{mckague2010generalized} suggests such an extension, and Ref.~\cite{vsupic2020self} describes this move as ``in line with the general spirit of self-testing in which one aims to certify the measurements up to all the intrinsic limitations of the device-independent scenario.''  

Our systematic resource-theoretic approach to self-testing, together with our arguments for LOSR as the appropriate set of free operations, eliminate any ambiguity regarding the status of chiral states: they {\em cannot} be self-tested.  Specifically, for all $B$, if $\ket \psi$ is chiral and $\ket{\psi} \mapsto B $, then the complex conjugate $\ket{\psi^*}$ also maps to $B$, that is, $\ket{\psi^*} \mapsto B $, and yet $\ket{\psi^*} \not\mapsto \ket{\psi}$.  Here, all conversion relations are relative to LOSR. Hence, by Definition~\ref{defnselfteste}, $\ket{\psi}$ cannot be self-tested. The set of free operations in the definition of self-testing cannot sensibly be modified ad hoc to change this fact.

Nonetheless, our work makes clear that there {\em is} a sensible way of extending the notion of self-testing to handle examples like that of chiral states, a way which is also in line with the aim of certifying processes `up to all the intrinsic limitations' of the scenario. In particular, one can define a {\em relaxed} notion of self-testing wherein one certifies (up to upward closure) a {\em set} of equivalence classes of states, rather than a unique equivalence class. Here, conversions are still evaluated relative to LOSR, so self-testing statements according to such a relaxed notion can still be understood as certifying the nonclassicality of a state based purely on the observed box. The price one pays for using this relaxed notion is simply that the precision with which one certifies the entangled state is diminished. This relaxation is described in more detail in the discussion surrounding Eq.~\eqref{eqrhoUC2} in Appendix~\ref{equivdefnself}.

\section{The resource theory of LOSR-entanglement}\label{sec:RTLOSR}

Having motivated the need for understanding LOSR-entanglement, we now prove some basic results about the preorder of bipartite and multipartite entangled states under LOSR. 

We begin by formally defining the sets of Local Operations (LO), Local Operations and Shared Randomness (LOSR), and Local Operations and Classical Communication (LOCC).
First, recall that the most general quantum operation taking an $n$-partite system described by the Hilbert space $\bigotimes_{i=1}^{n} \mathcal{H}^{(i)}$ to one described by the Hilbert space $\bigotimes_{i=1}^{n} \mathcal{K}^{(i)}$ is a completely-positive trace-preserving linear map between the corresponding operator spaces, that is, $\mathcal{E}: \mathcal{L}\big( \bigotimes_{i=1}^{n} \mathcal{H}^{(i)}\big) \to  \mathcal{L}\big( \bigotimes_{i=1}^{n} \mathcal{K}^{(i)}\big)$ where $\mathcal{L}(\mathcal{H})$ denotes the set of linear operations on  $\mathcal{H}$. 
A {\em local operation} is one that factorizes across the $n$-partition, 
$\mathcal{E}=
\mathcal{E}^{(1)}\otimes \cdots \otimes \mathcal{E}^{(n)}$
 where  $\mathcal{E}^{(i)} :  \mathcal{L}(\mathcal{H}^{(i)}) \to \mathcal{L}(\mathcal{K}^{(i)})$, and LO is the set of all such operations.  The operations in LOSR  are all and only those that can be expressed as a convex mixture of local operations, $\mathcal{E}=\sum_{\alpha} w_{\alpha} \mathcal{E}_{\alpha}^{(1)}\otimes \cdots \otimes \mathcal{E}_{\alpha}^{(n)}$,
   where $\{w_{\alpha}\}_{\alpha}$ denotes a probability distribution.  
To model classical communication between the parties, recall that a local nondestructive quantum measurement with an outcome labelled $j$ can be represented as a set of completely positive~\cite{NielsenAndChuang,PhysRevA.100.022112} trace-nonincreasing maps, $\{ \mathcal{E}_j\}_j$, whose sum, $\sum_j  \mathcal{E}_j$,  is a completely positive trace-preserving map.  An operation describing a single round of communication from Alice to Bob is described as follows: Alice implements a local quantum measurement represented by $\{ \mathcal{E}^{(A)}_j\}_j$, communicates the outcome $j$ to Bob, and then Bob implements a local quantum measurement, the choice of which may depend on $j$,  $\{ \mathcal{E}^{(B)}_{k|j}\}_k$, such that the overall operation is $\sum_{j,k} \left(  \mathcal{I} \otimes \mathcal{E}^{(B)}_{k|j} \right) \circ \left( \mathcal{E}^{(A)}_j \otimes \mathcal{I} \right)$.  An LOCC operation on an $n$-partite system can be defined as any operation which is a composition of such operations (but where the communication is between any pair of parties). 

Because LOCC allows for shared randomness to be generated by one party and then shared with the other parties by classical communication, it strictly includes LOSR,
    \begin{align}
        \label{eq:inclusion}
        \LOSR \subset \LOCC.
    \end{align}
It follows, therefore, that if a conversion from a state $\rho$ to a state $\sigma$ (mixed or pure) is not possible by LOCC, then it is also not possible by LOSR either:
\begin{align}
    \rho \not\mapsto \sigma\;{\rm by}\; \LOCC \ \ &\implies\ \  \rho \not\mapsto \sigma \;{\rm by}\; \LOSR.\label{eq:LOCCstrongestcontrapositive}
\end{align}
Consequently, (i) any witness of LOCC-nonconvertibility is a witness of LOSR-nonconvertibility as well,
(ii) any monotone relative to LOCC~\cite{HHH99b,vidal2000entanglement,GConcurrence,datta2009min} is also a monotone relative to LOSR, and (iii) the LOCC-equivalence classes contain the LOSR-equivalence classes.
 
However, there are instances of LOSR-nonconvertibility that are not instances of LOCC-nonconvertibility, or equivalently, cases where LOCC allows a conversion that is not allowed under LOSR.
Consequently, (i) there are witnesses of LOSR-nonconvertibility that are not witnesses of LOCC-nonconvertibility, (ii)
 there are LOSR monotones that are not LOCC monotones,
so that 
 a complete set of LOCC-entanglement monotones does \emph{not} constitute a complete set of LOSR-entanglement monotones,
and (iiii) it is {\em a priori} possible that some or all of the LOSR-equivalence classes of states are {\em strict} subsets of the corresponding LOCC-equivalence class of states.
Despite the {\em a priori} possibility of a distinction between LOSR-equivalence classes and LOCC-equivalence classes, for the case of pure states at least, one can show that this possibility is not realized.
Here, we make reference to the set of {\em Local Unitaries} (LU), which are simply local operations that are unitary.

\begin{lem}
    \label{lem:equivalence}
    Consider the equivalence relation associated to reversible     interconvertibility.\footnote{One resource is said to be {\em reversibly interconvertibl}e with another if conversions in both directions are possible under the free operations.  Note that the free operation achieving the conversion in one direction need not be the inverse of the free operation achieving the conversion in the opposite direction. }
    The following statements about equivalence of the $n$-partite pure states $|\psi\rangle$ and $|\phi\rangle$  (denoted $|\psi\rangle \leftrightarrows |\phi\rangle$) are equivalent:
    \begin{compactenum}[(i)]
        \item  $|\psi\rangle \leftrightarrows |\phi\rangle$ relative to LOSR.
        \item  $|\psi\rangle \leftrightarrows |\phi\rangle$ relative to LO.
        \item  $|\psi\rangle \leftrightarrows |\phi\rangle$ relative to LU.
        \item  $|\psi\rangle \leftrightarrows |\phi\rangle$ relative to LOCC.
    \end{compactenum}
\end{lem}

\begin{proof}
We note that the different sets of operations 
 form a chain of strict inclusions:
    \begin{align}
        \label{eq:inclusion_chain}
        \LU \subset \LO \subset \LOSR \subset \LOCC.
    \end{align}
Next, we appeal to the known result
  that two pure states $|\psi\rangle$ and $|\phi\rangle$
     are LOCC-equivalent if and only if they are LU-equivalent \cite{bennett2000exact}. Since the set-theoretic inclusion relations are such that LO and LOSR are both between LU and LOCC in Eq.~\eqref{eq:inclusion_chain}, the result follows. 
\end{proof}

It follows from Lemma~\ref{lem:equivalence} and the discussion above that {\em all} of the differences between the LOCC and LOSR preorders on pure states correspond to pairs of equivalence classes which are strictly ordered under LOCC but are incomparable under LOSR.

The LOCC preorder over pure states has been completely characterized in the bipartite case~\cite{nielsen1999conditions}, and in the following we obtain the analogous result for the LOSR preorder.  We also express some results concerning the  $n$-partite case for $n>2$.

\begin{samepage}
\begin{lem}
    \label{lem:conversion}
    The following statements about interconversion between $n$-partite pure states $|\psi\rangle$ and $|\phi\rangle$ are equivalent 
    \begin{compactenum}[(i)]
        \item  $|\psi\rangle \mapsto |\phi\rangle$ by LOSR.
        \item  $|\psi\rangle \mapsto |\phi\rangle$ by LO.
        \item  $\exists |\eta_1\rangle \dots |\eta_n\rangle, |\zeta\rangle : |\psi\rangle |\eta_1\rangle \dots |\eta_n\rangle \leftrightarrows
          |\phi\rangle |\zeta\rangle$ by LU, as illustrated in Fig.~\ref{main_Lemma_three_parties}.
    \end{compactenum}
\end{lem}
\end{samepage}

\begin{figure}
\centering
    \includegraphics[scale=1]{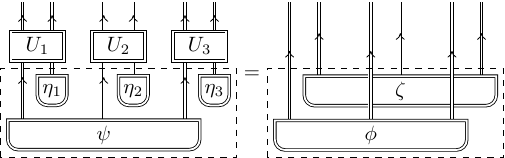}
    \caption{A depiction of condition (iii) of Lemma~\ref{lem:conversion} for $n = 3$ parties. 
    Note that for each $U_i$, the factorization of its input space into a pair of subsystems need not be the same as the factorization of its output space.}
    \label{main_Lemma_three_parties}
\end{figure}

\begin{proof}
    (ii) $\implies$ (i) holds trivially, because every $\LO$ operation is an $\LOSR$ operation. (i) $\implies$ (ii) follows, because every $\LOSR$ operation is a mixture of $\LO$ operations, and given that the state $\ket{\phi}$ is pure, it follows that
     the states resulting from each local operation of this mixture must be proportional to $|\phi\rangle$. Consequently, any one of the $\LO$ operations in the mixture can convert $|\psi\rangle$ to $|\phi\rangle$ on its own.

    (iii) $\implies$ (ii) trivially, because preparing local auxiliary systems, applying local unitaries, and taking partial traces are free $\LO$ operations. (ii) $\implies$ (iii) follows from the Stinespring dilation theorem~\cite{Stinespring1955,Paulsen2003Apr}: any $\LO$ operation is an $n$-partite tensor product of local trace-preserving completely positive operations, and each of these 
    can be implemented by introducing an auxiliary system prepared in a pure state $|\eta_i\rangle $, coupling this auxiliary system to the system of interest by a unitary $U_i$, then tracing out some subsystem of this composite.
  Given that the final reduced state on the output $n$-partite system must be the pure state $|\phi\rangle$, it follows that a subsystem of the output composite can only be traced out if it is unentangled with the $n$-partite system.  Consequently, the final joint state of the output $n$-partite system and the systems that are traced out
  must be of the form $|\phi\rangle \otimes |\zeta\rangle$ for some (possibly entangled) ``junk'' state $|\zeta\rangle$.
\end{proof}

Evidently, Lemmas~\ref{lem:equivalence}~and~\ref{lem:conversion} provide avenues for understanding the LOSR-entanglement preorder on pure states by utilizing known results about LU-equivalence~\cite{Kraus2010Local,liu2012local,Acin2001puretriclassify,Barnum2001,biamonte2013tensor}. 

For instance, in the bipartite case, it is well known that local unitary equivalence between pure  states has a particularly simple form: if $\lambda_{\psi}$ denotes the vector of squared Schmidt coefficients of a state $|\psi\rangle$, and $v^{\downarrow}$ denotes the permutation of a vector $v$'s components such that they are ordered from largest to smallest, then $\ket \psi$ and $\ket \phi$ are LU-equivalent if and only if 
\begin{align}
    \label{eq:bipartite_LUequivalence}
     \lambda_{\psi}^{\downarrow} = \lambda_{\phi}^{\downarrow},
\end{align}
that is, if and only if their vectors of squared Schmidt coefficients are equal up to permutation.
Lemma~\ref{lem:equivalence} then implies that this is also the condition for LOSR-equivalence between pure states.  

This characterization of the LU-equivalence of bipartite pure states also implies that condition (iii) in Lemma~\ref{lem:conversion} has a particularly simple form\footnote{The equivalence of condition (ii) and this form of condition (iii) for the bipartite case is stated as Exercise 12.22 of Ref.~\cite{NielsenAndChuang}.}, resulting in the following corollary concerning LOSR-convertibility:
\begin{cor}
    \label{cor:bipartite}
    A bipartite pure state $\ket{\psi}$ can be converted to a bipartite pure state $\ket{\phi}$ by LOSR if and only if
    \begin{align}
    \label{eq:bipartite_extension}
    \exists \ket \zeta: \lambda_{\psi}^{\downarrow} = (\lambda_{\phi} \otimes \lambda_{\zeta})^{\downarrow}.
    \end{align}
 \end{cor}

It follows that the pair of states $|\psi_{\rm max}\rangle$ and $|\psi_{\rm partial}\rangle$ from the main text, which have the same Schmidt rank but different vectors of Schmidt coefficients, are such that neither converts to the other under LOSR, i.e., they are LOSR-incomparable. This justifies Eq.~\eqref{notconvmaxtopartial}.

In the case of $n$-partite systems for $n>2$, a necessary condition for LU-equivalence between pure states is
the equality of squared Schmidt coefficients for each bipartition $\beta$ of the $n$-partite system, but (unlike the bipartite case) this is not a sufficient condition~\cite{Kraus2010Local,Acin2001puretriclassify}. 
Nevertheless, combining this condition with Lemma~\ref{lem:conversion} yields the following corollary concerning LOSR-convertibility.

\begin{cor}
    \label{cor:multipartite_extension}
    An $n$-partite pure state $\ket{\psi}$ can be converted to an $n$-partite pure state $\ket{\phi}$ by LOSR only if
    \begin{align}
        \label{eq:multipartite_extension}
\exists \ket \zeta, \forall \beta : (\lambda^{(\beta)}_{\psi})^{\downarrow} = (\lambda^{(\beta)}_{\phi} \otimes \lambda^{(\beta)}_{\zeta})^{\downarrow},
    \end{align}
    where for a pure state $\ket{\omega}$, $\lambda^{(\beta)}_{\omega}$ denotes the vector of its squared Schmidt coefficients 
    with respect to bipartition $\beta$ of the $n$-partite system.
\end{cor}

Despite the insufficiency of Corollary~\ref{cor:multipartite_extension} for determining LOSR-convertibility (which we demonstrate with an explicit example in Appendix~\ref{insufficiency}), it provides a necessary condition for LOSR-convertibility which
 is {\em not} a necessary condition for LOCC-convertibility.
  As an example, by appealing to Corollary~\ref{cor:multipartite_extension}, we show in Appendix~\ref{BellGHZ} that $\ket{\psi_{\rm 2Bell}}$ and $\ket{\psi_{\rm GHZ}}$ (which are strictly ordered relative to LOCC) are incomparable relative to LOSR, thereby justifying Eq.~\eqref{noconvBellGHZ}. 

The condition stipulated by Corollary~\ref{cor:multipartite_extension}, that there exists a pure state $\ket{\zeta}$ satisfying Eq.~\eqref{eq:multipartite_extension},
is a non-trivial requirement. Indeed, the problem of deciding whether there exists an $n$-partite pure state $\ket{\zeta}$ admitting a given set of vectors of squared Schmidt coefficients $\lambda^{(\beta)}_{\zeta}$ for a given
 family of bipartitions $\beta$ is known as the spectral quantum marginals problem~\cite{klyachko2006quantum,walter2013entanglement}.
Appendix~\ref{alg} describes a simple method for computing  the set of vectors $\lambda_{\zeta}^{(\beta)}$ that satisfy Eq.~\eqref{eq:multipartite_extension}, given $\lambda_{\psi}^{(\beta)}$ and $\lambda_{\phi}^{(\beta)}$, which provides the input to the spectral quantum marginals problem.

We also note a second corollary of Lemma~\ref{lem:conversion}.
 Recall that the Schmidt rank of a state $\ket{\psi}$ with respect to a bipartition $\beta$, henceforth denoted $\mathrm{SR}_{\psi}^{(\beta)}$, is defined as the 
number of non-zero entries of the vector of squared Schmidt coefficients $\lambda_{\psi}^{(\beta)}$.

    \begin{cor}
        \label{cor:equivalent_or_incomparable}
        If two pure states $\ket{\psi}$ and $\ket{\phi}$ have the same Schmidt ranks for all bipartitions, i.e., 
        \begin{align}
            \label{eq:equal_schmidt_ranks}
            \forall \beta: \mathrm{SR}_{\psi}^{(\beta)} = \mathrm{SR}_{\phi}^{(\beta)},
        \end{align}
        then they are either 
        LOSR-equivalent (equivalently, LU-equivalent)
         or else LOSR-incomparable.
    \end{cor}
    \begin{proof}
Suppose that $\ket{\psi} \mapsto \ket{\phi}$ under LOSR, so that condition (iii) from Lemma~\ref{lem:conversion}
         is satisfied.           Under the assumption that 
        Eq.~\eqref{eq:equal_schmidt_ranks} holds, it follows that the state $\ket{\zeta}$ in this condition must be such that $\forall \beta: \mathrm{SR}_{\zeta}^{(\beta)} = 1$.  But this implies that $\ket{\zeta}$ is a product state, $\ket \zeta = \ket {\zeta_1} \cdots \ket {\zeta_n}$, in which case 
condition (iii) from Lemma~\ref{lem:conversion} reduces to $\exists |\eta_1\rangle \dots |\eta_n\rangle, |\zeta_1\rangle \dots |\zeta_n\rangle : |\psi\rangle |\eta_1\rangle \dots |\eta_n\rangle \leftrightarrows |\phi\rangle |\zeta_1\rangle \dots |\zeta_n\rangle$ by LU,
         which is a condition that is symmetric in $\ket \psi$ and $\ket \phi$.  The same logic holds if one supposes that $\ket{\phi} \mapsto \ket{\psi}$ under LOSR. Consequently, either both conversion relations hold or neither does.  In other words, either $\ket{\psi}$ and $\ket{\phi}$ are LOSR-equivalent  or they are LOSR-incomparable.  Finally, LOSR-equivalence is the same as LU-equivalence by Lemma~\ref{lem:equivalence}.\footnote{By the structure of this proof, it is clear that one can replace the assumption of sameness of Schmidt ranks for all bipartitions with
  any other condition that is symmetric in $\ket \psi$ and $\ket \phi$ and which establishes that $\ket \zeta$ must be a product state. For instance, if two pure states $\ket \psi$ and $\ket \phi$ have the same marginal entropies for each subsystem (i.e., are ``marginally isentropic'') then one can also conclude that $\ket{\zeta}$ is a product state and consequently that $\ket \psi$ and $\ket \phi$ are either LOSR-equivalent or LOSR-incomparable.  This result is similar to Theorem 1 of Ref.~\cite{bennett2000exact}, which concerns LOCC rather than LOSR, and which establishes that if $\ket \psi$ and $\ket \phi$ are marginally isentropic, then they are either LU-equivalent or LOCC-incomparable.   Because LOCC-incomparability implies LOSR-incomparability, and LU-equivalence is the same as LOSR-equivalence (by Lemma~\ref{lem:equivalence}), the analogue of Corollary~\ref{cor:equivalent_or_incomparable} for marginally isentropic states can be obtained as a corollary of Theorem 1 of Ref.~\cite{bennett2000exact}.}
    \end{proof}

Applying Corollary~\ref{cor:equivalent_or_incomparable} to the case of bipartite states, it follows that any pair of pure bipartite states having equal Schmidt rank but different squared Schmidt coefficients are incomparable under LOSR.  Thus, for example, it follows that all two-qubit entangled pure states are LOSR-incomparable.\footnote{This provides a second way of seeing that the pair of bipartite pure states $|\psi_{\rm max}\rangle$ and $|\psi_{\rm partial}\rangle$ are LOSR-incomparable.}

Another result about LOSR-entanglement for pure bipartite states that follows from Lemma~\ref{lem:conversion}
concerns entanglement catalysis.  Specifically, whereas there is nontrivial catalysis for pairs of bipartite pure states in the theory of LOCC-entanglement~\cite{JonathanPlenio}, this is not the case in the theory of LOSR-entanglement. If $\ket \psi$ cannot be converted to $\ket \phi$ under LOSR, i.e., $\ket{\psi} \not\mapsto \ket{\phi}$, then $\ket \psi$ also cannot be {\em catalytically} converted to $\ket \phi$ under LOSR, i.e., there is no pure state $\ket \chi$ such that $\ket{\psi} \otimes \ket{\chi} \mapsto \ket{\phi} \otimes \ket{\chi}$.  The proof is provided in Appendix~\ref{catalysis}. Consequently, whether one uses LOCC or LOSR as the free operations also makes a difference for the catalytic order over entangled states.

It is also worth noting that one can easily construct LOSR-entanglement monotones that are not also LOCC-entanglement monotones.  We already noted that for every instance of the anomaly of nonlocality, one can define an LOSR-entanglement monotone in terms of the yield of nonlocal box appearing in that anomaly (see Appendix~\ref{monotones} for details). Though conceptually natural, these monotones are likely to have limited practical utility unless one can provide a general solution to the optimization that is part of their definition.  If the set of states on which monotonicity is required to hold is taken to be the set of pure states rather than the full set of states,  i.e., if one is content to have a pure-state monotone rather than a true monotone, then it is not difficult to construct examples of pure-state LOSR-entanglement monotones that are not also pure-state LOCC-entanglement monotones and for which a closed-form expression can be given. For instance, the quantity known as {\em entanglement spread}~\cite{Harrow2010} is such a monotone, as is every function in the 2-parameter family of the form $\Delta_{\alpha \beta}(\psi) = S_{\alpha}(\psi) - S_{\beta}(\psi)$, where $0\le \alpha < \beta\le \infty$,
 and where $S_{\alpha}$ is the order-$\alpha$ Renyi entropy (entanglement spread is the $\alpha=\infty$, $\beta=0$ member of this family). This follows from a result in Ref.~\cite{Hayden2003}; see Theorem 1 and Remarks 1 and 2 therein.\footnote{We thank Patrick Hayden for bringing this fact to our attention.}  Other pure-state LOSR-entanglement monotones for pure states can easily be constructed using our characterization of the LOSR order in Lemma~\ref{lem:conversion} and Corollary~\ref{cor:multipartite_extension}.

\section{The ordering of entangled states depends on the network structure}\label{NetworkStructure}

As stipulated in the introduction, we take the entanglement properties of multipartite quantum states to be {\em their nonclassical correlational properties}. 
  The question of 
what counts as nonclassicality and how to order states by their degree of nonclassicality,
 in turn,  can only be answered relative to a choice of network structure between the parties, that is, a specification of 
 the particular communication channels and common sources that link them.\footnote{Note that the specification of whether the channels and sources are classical or quantum  is not considered to be part of the specification of the network structure.  In our usage, the network structure refers merely to the pattern of connectivity, and does not distinguish classical and quantum. } (See Ref.~\cite{wolfe2020quantifying} for a discussion of this point.)

Consider first how this claim is true if one focuses on boxes, rather than states.
Imagine a box for which the conditional probability distribution of its outputs given its inputs violates a Bell inequality.  Should the correlational properties of such a box be considered nonclassical?   If the box is realized in a network wherein the parties only have access to a common source (Fig.~\ref{fig:common_source_networks}), then it should, because in this context, such violations certify the nonclassicality of the source~\cite{Wood2015}.  If, however, the box is realized in a network wherein the parties have communication channels  (Fig.~\ref{fig:two_way_channels})\footnote{And no restriction is stipulated on the spatio-temporal relations between the inputs and outputs of the box---see discussion below.}, then it should not, because any Bell inequality violation is easily simulated using classical communication between the wings.

\begin{figure}[htb]
\centering
    \subfloat[A common source shared by all parties.]{
        \makebox[0.48\textwidth][c]{
            \raisebox{-0.5\height}{\includegraphics[scale=1]{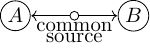}}
            \raisebox{-0.5\height}{\includegraphics[scale=1]{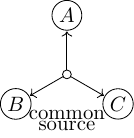}}
        }
        \label{fig:common_source_networks}
    }\\
    \subfloat[Two-way channels between all pairs of parties.]{
        \makebox[0.48\textwidth][c]{
            \raisebox{-0.5\height}{\includegraphics[scale=1]{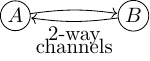}}
            \raisebox{-0.5\height}{\includegraphics[scale=1]{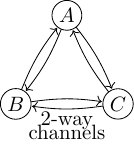}}
        }
        \label{fig:two_way_channels}
    }
    \caption{Common source and two-way channel networks involving two and three parties.
     }
     \label{fig:common_source_and_two_way}
\end{figure}

The same sorts of considerations apply to quantum states: whether the correlational properties of a particular state should be considered nonclassical and how to quantify that nonclassicality are questions whose answers depend
 on the network.  This is less obvious than in the case of boxes because, for states, the difference between the two networks in Fig.~\ref{fig:common_source_and_two_way}
 is not manifest in the classical-nonclassical boundary.  The difference does show up, however, in the {\em ordering} over the states.
 In particular, for the networks in Fig.~\ref{fig:common_source_networks}, the ordering of multipartite states is determined by LOSR-conversion relations, while for the networks in Fig.~\ref{fig:two_way_channels}, it is determined by LOCC-conversion relations. 
 The results of Section~\ref{sec:RTLOSR} show that these orderings are indeed different, so that one has different notions of entanglement  in the two networks, which we have termed LOSR-entanglement and LOCC-entanglement respectively. 

 Although the location of the classical-nonclassical boundary for states does not differ for the pair of network structures in Fig.~\ref{fig:common_source_and_two_way}, in other pairs of networks it can.  We provide an example in Appendix~\ref{ExamplePairStructures}.

To summarize, then: understanding the nonclassicality of the correlational properties of a given type of process involves understanding the  gap between the correlational properties that are classically realizable and those that are quantumly realizable. But given that what is realizable (either classically or quantumly) differs between network structures, the nature of such gaps---and thus the nature of the nonclassicality of correlational properties---is relative to a network structure. 
 Ultimately, the network-relativity of assessments of nonclassicality is due to the fact that whether correlations are realizable or not  is dependent on the {\em causal relations} that exist among the parties~\cite{Wood2015,schmid2020unscrambling}, and the network structure stipulates these causal relations.

 Much of this article has concerned the interplay of entanglement and nonlocality in networks of the form of Fig.~\ref{fig:common_source_networks}, with common sources among the parties.  We now turn to the question of how this interplay appears in networks of the form of Fig.~\ref{fig:two_way_channels}, with channels among the parties.

As we noted in Sec.~\ref{sec:unifiedRT}, it is not obvious a priori that it even makes sense to study nonlocality in such a network given that any Bell inequality violation can be simulated by classical communication. 
\rob However, there is an approach that allows such a possibility.  In this approach, one changes the resources under consideration from the boxes defined in Sec.~\ref{sec:introduction} to generic multipartite processes with classical inputs and outputs, supplemented by {\em a specification of the spatio-temporal relations between all of the inputs and outputs}.  To distinguish these from the boxes defined in Sec.~\ref{sec:introduction}, we refer to these as {\em spatio-temporally-indexed boxes}.  One furthermore restricts attention to the subset of spatio-temporally-indexed boxes
  for which the settings and outcomes of a given party
 are {\em space-like separated} from those of the other parties.  We term these ``space-like boxes''. 
For this special subset of processes,
   the internal causal structure can only be that of a common cause, just like the notion of box defined in Sec.~\ref{sec:introduction}.  
Furthermore, focussing on the subset of spatio-temporally-indexed boxes corresponding to 
space-like boxes allows Bell inequality violations to be indicative of nonclassicality even in the presence of a communication channel.  The reason is that \blk although the parties have access to classical communication channels, these channels {\em simply do not help} with achieving  conversions between one space-like box and another because any processing at one wing must be implemented at space-like separation from a processing at the other wings in order for the overall processing's output to also be a space-like box.   \rob
It follows that to assess conversions {\em among space-like boxes}, it is sufficient to consider just LOSR operations.
 In other words, although the free operations in a resource theory of spatio-temporally-indexed boxes may be taken to be LOCC, 
  for space-like boxes, a conversion is achievable by LOCC
   if and only if it is achievable by LOSR. \blk

Meanwhile, for achieving {\em state-to-box} and {\em state-to-state} conversions, the classical communication channel {\em is} useful, so these conversions \rob must be evaluated relative to LOCC in this alternative resource theory.\blk Consider the state-to-box case, for example.  If the entangled quantum state can be prepared prior to the time at which the setting variables of the \rob space-like \blk box are fixed, then there is no reason to assume that local operations on one part of the entangled state need to be space-like separated from those on another part.  Consequently, measurements can be performed on the different parts of the entangled state and the results can be communicated among the parties and used to control subsequent local operations.\footnote{Some authors have suggested a set of free operations that is distinct from LOCC and termed Wirings with Prior-to-Input Classical Communication (WPICC)~\cite{WPICC}.  The latter allows only classical communication prior to receiving the values of the input variables.  It is not clear, however, what sort of physical or experimental restriction might enforce such a prohibition on the use of a classical channel.} 

This approach provides an alternative means of resolving 
 the inconsistencies described in Secs.~\ref{sec:anomalies} and \ref{sec:multipartite}.  
 For instance, in the inconsistency associated with the anomaly of nonlocality,
 the claim that $\ket{\psi_{\rm max}} \not\mapsto B$ (see Eq.~\eqref{notconv}), which is true for LO, is not true for LOCC, because under LOCC, $\ket{\psi_{\rm max}}$ can be converted to $\ket{\psi_{\rm partial}}$ which can be converted to $B$. So this approach, which reduces mathematically to LOSR for box-to-box conversions but LOCC for state-to-state and state-to-box conversions, implies replacing Eq.~\eqref{notconv} with its negation, while preserving Eqs.~\eqref{convmaxtopartial} and ~\eqref{convpartialtobox}, and therefore provides a
  way out of the inconsistency.

The restriction to space-like boxes
may at first appear quite natural.  After all, in the context of Bell experiments, one seeks to ensure that the outcome and setting on one wing cannot have a causal influence on the outcome at any other wing 
  (i.e., that the locality loophole is sealed).  
Sealing loopholes in a Bell experiment is done for the purpose of ruling out theories---typically described as ``locally realist''---which make alternative predictions to quantum theory.
  If one is not in the business of trying to rule out such alternatives, 
 but rather in the business of leveraging
  quantum resources for information-processing tasks (while presuming
   the correctness of quantum theory), then one need not
     restrict attention to boxes whose wings are space-like separated.  One might also consider
      boxes for which one is confident, based on the setup of the experiment, that there is no causal influence between the wings.   The wings being space-like separated is a sufficient but not a necessary condition for satisfying this no-causal-influence condition.
 
Nonetheless, 
 it may be the case that the information-processing task in question is of a cryptographic nature, and that the users of the protocol are explicitly worried about an adversary that is attempting to make it appear  that the no-causal-influence condition is satisfied when it is not, in order to fool them into thinking that they are achieving a cryptographically secure protocol.  In such circumstances, it may be appropriate to require space-like separation as a guarantee of the no-causal-influence condition being satisfied.

  \sloppy Such circumstances, therefore, motivate the study of the interplay of LOCC-entanglement and nonlocality, that is, the interplay of the nonclassicality of states and \rob space-like \blk boxes in the network with channels among the parties (Fig.~\ref{fig:two_way_channels}).
Note, however, that there are significant complications in undertaking
 such a study within a formal resource-theoretic framework.
\rob
 One such complication is that 
 the formal representation of a resource 
  includes
 a specification of the spatio-temporal relations between all of its inputs and outputs, unlike the notion of box introduced in Sec.~\ref{sec:introduction} and studied in this article.
Being able to forego the complication of requiring such a specification is one of the advantages of our approach.
\blk

After a draft of this article appeared (and partly in response to it), Ref.~\cite{Gour2020RT} has argued for 
\rob the alternative approach wherein one considers free operations including 2-way communications (Fig.~\ref{fig:two_way_channels}) \blk
 and has begun the project of developing a formal resource theory of this type.\footnote{Whereas we have here argued that the appropriate notion of nonclassicality for states and boxes is relative to the network connecting the parties, with LOCC and LOSR emerging as the appropriate free operations in different circumstances, Ref.~\cite{Gour2020RT} argues against the LOSR paradigm.  It claims that \enquote{LOCC is needed to fully reveal the nonlocal features of quantum states.}  In the next section, however, we explain why whether a state  is judged to be entangled or not should {\em not} be defined in terms of whether it can generate a nonlocal box.  Furthermore, even if one {\em were} to adopt such a definition, it is not the case that the only way to witness the entanglement of a state is by activating it in a protocol that makes use of  LOCC operations.  Indeed, in the next section, we summarize the many ways that have been discovered to witness entanglement by converting states to processes with classical inputs and outputs using LOSR. 
Therefore, although we do not dispute that LOCC operations can sometimes be appropriate for the study of the interplay of entanglement and nonlocality (namely, if the network structure includes communication channels),
the arguments in Ref.~\cite{Gour2020RT} against the appropriateness of LOSR operations are unconvincing.
}

It is also worth noting that the results obtained in such a setting will clearly 
be at odds with conventional notions about the interplay of entanglement and nonlocality.  
For example, in the LOCC paradigm, Hardy-type proofs of Bell's theorem can be achieved not only with partially entangled states, but with maximally entangled states as well (the negation of Eq.~\eqref{notconv} discussed above), contrary to the conventional claim.
It is also the case that within this approach, it is impossible to uniquely certify all partially entangled states of a given Schmidt rank by self-testing, again contrary to the usual result. 
Finally, the fact that there is no additional ``3-way communication resource'' that goes beyond having
2-way communication channels between every pair of parties
 implies that, in the LOCC paradigm, there is no obvious candidate for a notion of genuine 3-way entanglement and consequently no natural counterpart to a notion of genuine 3-way nonlocality,  as we noted in  Sec.~\ref{sec:multipartite}.

In any case, what is clear is that the  interplay between entanglement and nonlocality can be markedly different in different network structures and therefore it is critical to articulate which network structure one is assuming when making any claims about the nature of entanglement, nonlocality, or their interplay. 
This realization is one of the key contributions of our work.

\section{On entangled states that do not violate any Bell inequality}\label{Werner}

The fact that there are mixed entangled states that do not violate any Bell inequality~\cite{werner1989quantum,Barrettlocal}  is often taken as a surprising feature of the interplay between the entanglement of states and the nonlocality of boxes.  We here argue that, in the resource-theoretic approach advocated in this article, this fact is not surprising.  Indeed, it is an instance of a phenomenon that is generic among resource theories.  The conclusion which has sometimes been drawn from this phenomenon---that certain entangled states cannot be certified device-independently---is now known to be false and can be seen, in retrospect, to be an artifact of considering a limited type of 
network structure
in entanglement certification protocols.

In the framework of resource theories as partitioned process theories, outlined in Ref.~\cite{coecke2016mathematical}, the distinction between free and nonfree processes (of any type) is part of the definition of the resource theory. It cannot be stipulated arbitrarily, because the free processes are required to be closed under arbitrary composition operations, such that they form a subtheory of the enveloping process theory.  The question of what sorts of resources of a given type, $T'$, one can obtain from a given resource $R$ of type $T$ is an interesting one (studied in Refs.~\cite{semiquantum,rosset2020characterizing}), but it does not serve to {\em define} whether $R$ is to be considered free or not.  
It is worth particularizing this claim to the case of interest here.  We are concerned with processes of state type and of box type wherein being free corresponds to being  realizable  using only classical sources.  In this context, whether or not a given state can be converted to a nonlocal box is certainly an interesting question, but it is {\em not} a necessary condition for the state to be deemed nonclassical.

Nonetheless, imagine a sceptic who tries to maintain that the classical-nonclassical boundary for quantum states {\em should} in fact be taken to be the boundary between states that yield only local boxes and those that can yield nonlocal boxes, and that nonclassicality of states should be quantified in terms of features of this yield rather than the ordering of the states. 
Such a position might be motivated by the idea that 
the only nonclassical feature of an entangled quantum state is its ability to violate a Bell inequality. Nonetheless, we now explain why we believe this to be an ill-conceived approach.

Consider the possibility of trying to define a free-nonfree distinction for processes of type $T$ in terms of whether or not they can yield a nonfree process of type $T'$.  Note, first of all, that to get off the ground with this sort of definition, one must already have in hand a free-nonfree distinction for processes of type $T'$.  But then, whatever physical considerations led to the free-nonfree distinction in the type $T'$ sphere, one can (and should) avail oneself of the 
  same considerations to obtain a free-nonfree distinction in the type $T$ sphere. To specialize to the case of interest here: it is uncontroversial that the classical-nonclassical distinction for boxes should be defined in terms of realizability by classical shared randomness, which corresponds to whether all of the Bell inequalities are satisfied or not.  But then, the classical-nonclassical distinction for states can (and should) {\em also} be defined directly in terms of realizability by classical shared randomness (i.e., LOSR), which corresponds to whether it is separable or not.

Note that the imperative to define 
 the free subtheory based on a physical restriction, outlined in Sec.~\ref{sec:unifiedRT}, supports this approach insofar as it yields a type-independent notion of freeness.  In the case of interest here, the physical restriction is the classicality of all of the sources and channels in the network, 
  and it defines a free-nonfree distinction for states, boxes, and every other type~\cite{semiquantum,rosset2020characterizing} of multipartite process in that network. 

Another problem with the sceptic's proposal is that there are {\em many} different types of processes to which a type-$T$ process can be converted: type $T'$, type $T''$, $\dots$.
How does one decide which of these ought to be used to define nonfreeness of type-$T$ processes?

Indeed, this ambiguity arises in any attempt 
  to define the nonclassicality of states in terms of the nonclassicality of processes obtained from these. 
   Consider the device-independent certification of entanglement, i.e., certification by 
   conversion into a process with only classical inputs and outputs.  
    If we convert the state to a box,
    then certain nonseparable states yield only free boxes, so that if we defined the classical-nonclassical boundary for states in terms of such conversions, this boundary would lie {\em strictly within} the space of nonseparable states, with Werner states~\cite{werner1989quantum} being examples of nonseparable states that are nonetheless deemed to be classical.   
However, as soon as one considers converting states into other sorts of processes, it becomes possible to certify entangled states that could not be certified by boxes.  Examples of such processes
include: those depicted in Fig.~\ref{fig:filt_a}, obtained by implementing in each lab a multi-meter (i.e., a measurement whose identity is determined by a classical input) {\em after} a filtering operation on the quantum state (a nondestructive measurement whose identity is fixed)~\cite{popescu1995bell,gisin1996hidden}; those depicted in Fig.~\ref{fig:filt_b}, obtained by implementing in each lab a sequence of nondestructive multi-meters on the quantum state~\cite{gallego2014nonlocality}; those depicted in Fig.~\ref{fig:filt_c}, which implement 
local operations with a particular nontrivial causal structure {\em within} the laboratory of each party~\cite{bowles2018device}.

Alternatively, one can also certify the nonclassicality of arbitrary entangled states within the simple causal structure of the Bell scenario, if one is willing to relax one's notion of device-independence. This is the case, for example, if one considers local processings with classical outputs and secondary quantum inputs, the so-called semi-quantum scenario, pioneered by Buscemi~\cite{Buscemi2012LOSR} and depicted in Fig.~\ref{fig:filt_d}.

\begin{figure}[htb]
\centering
    \subfloat[]{
        \raisebox{-0.645\height}{\includegraphics[scale=1]{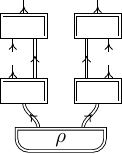}}
        \label{fig:filt_a}
    }
   \hspace{1.3em} \subfloat[]{
        \raisebox{-0.5\height}{\includegraphics[scale=1]{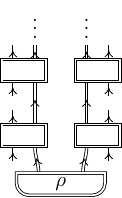}}
        \label{fig:filt_b}
    }\\
    \subfloat[]{
        \raisebox{-0.5\height}{\includegraphics[scale=1]{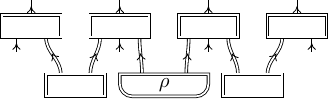}}
        \label{fig:filt_c}
    }
    \\
    \subfloat[]{
        \raisebox{-0.5\height}{\includegraphics[scale=1]{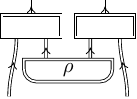}}
        \label{fig:filt_d}
    }
    \caption{
(a) A box with prior filtering of the state. 
(b) A box implementing a sequence of nondestructive measurements on the state.     (c) A process with only classical inputs and classical outputs, but constrained to have a particular causal structure in each party's lab.  Such a process can witness {\em all} entangled states device-independently. (d) A process (sometimes termed a semi-quantum channel) which is like a box, but with quantum rather than classical inputs.  Such a process can certify any entangled state within the standard Bell causal structure, but only device-dependently.}
     \label{fig:filtering_sequences_etc}
\end{figure}

So, whether one can certify (much less quantify) an entangled state or not by converting it into a process of a different type depends explicitly on the type. 

Note that our argument against the sceptic's position does not depend on the results of Ref.~\cite{bowles2018device},
that there {\em do} in fact exist processes with classical inputs and outputs that can certify all entangled states.  Even if there had not existed any such processes, this would not add credibility to the idea that the classical-nonclassical boundary is anything but the separable-nonseparable boundary.
If a process of one type is nonfree, there is no guarantee that in a conversion to another type of process, the resulting process
 is also nonfree.\footnote{The question of what kinds of type-changing transformations are in fact nonclassicality-preserving is studied in Ref.~\cite{semiquantum,rosset2020characterizing}.}
 In other words, in type-changing resource conversions, nonfreeness of the initial resource is not a sufficient condition for nonfreeness of the final resource.  
 
This insufficiency is obvious if the type of the final resource is too impoverished.  Consider, for instance, the type-changing conversion that maps a multi-partite entangled state to a multipartite process with only classical outputs (i.e., no inputs), represented by a multi-partite joint probability distribution.  Because {\em every} multi-partite joint probability distribution can be realized with a classical source shared among all of the parties, it follows that in the LOSR-ordering of resources, 
each such distribution is a free resource, even though every entangled state is nonfree.\footnote{Note that this is not a trivial example, since processes that have only classical outputs {\em can} witness entanglement in some networks. For instance, if one considers the triangle network of Fig.~\ref{fig:triangle_network} in Appendix~\ref{ExamplePairStructures},
there exist tripartite joint distributions, such as the one described in Eq.~\eqref{Fritzdistn}, that are {\em only} realizable quantumly, not classically, and hence the type-changing resource conversion that takes states to distributions {\em can} witness some entangled states device-independently within this network.}

Once one recognizes that there is no reason to expect nonfreeness to be preserved in type-changing conversions,
one sees that there is no reason to expect that all entangled states should be able to generate a nonlocal box.
 Without any justification for such an expectation, there is no reason to consider the phenomenon to be surprising. 
Indeed, with all of the recent progress in this area~\cite{gallego2014nonlocality,bowles2018device}, it has become clear that the interesting question is rather: what distinguishes the process types that can witness any entangled state from the process types that cannot?

\section{Discussion}\label{sec:Discussion}

 We have shown that by recasting the interplay of entanglement and nonlocality within a resource-theoretic framework,  several puzzling features 
 are seen to be an artifact of 
 not assuming a consistent choice of the free operations for all types of processes.   
 In the resource theory where the network is presumed to be one where the parties share a common source and where the {\em free} sources are the classical ones, the set of free operations is LOSR for {\em all} types of conversion relations and the relevant notion of entanglement is LOSR-entanglement.   This notion of entanglement is shown to have a particularly natural interplay with the notion of nonlocality.

In particular, we have shown that this perspective provides a particularly satisfactory resolution of the long-standing anomalies of nonlocality. 
   In addition to its foundational dividends on this front, the resource theory based on  LOSR operations is also likely to be significant for certain more practical problems.  For instance, the fact that partially entangled states provide an advantage over maximally entangled states for the tasks of certified randomness generation \cite{Acin2012randomnessvsnonlocality} and for the rate of secret key extraction \cite{Scarani2006QKD,Acin2006QKD} suggests that one should seek to identify the LOSR-entanglement monotones that quantify the degree of success one can achieve in such tasks.

Our suggestion to use LOSR supplemented by 2-way shared entanglement (LOSR2WSE) to define genuine 3-way entanglement and genuine 3-way nonlocality has not only conceptual merit---providing a satisfactory resolution to a tripartite anomaly of nonlocality for instance---but also the potential for practical applications.  For example, it opens the door to defining quantitative measures of these types of entanglement and nonlocality. 
This is in contrast to the LOCC-motivated notions of genuine 3-way entanglement and nonlocality, which we showed violate basic principles of a resource theory.  Our approach is also easily adapted to the study of types of nonclassicality in network structures beyond that of a common source shared among all parties
 ~\cite{Henson2014,Fritz2012beyondBell,Wolfe2016inflation}.

Self-testing of entangled states is another topic that has practical significance, insofar as it provides a 
  means of certifying quantum devices. It is crucial, therefore, that we properly understand 
 what, precisely, is certified in a self-testing experiment.
Our resource-theoretic reframing of self-testing brings additional clarity to this
 question, by establishing the freedom up to which a state may be identified by self-testing.
pure bipartite entangled states of a fixed Schmidt rank).
In addition, our new approach shows that certain mixed states can be self-tested.
This sets the stage for investigations of the possibility of self-testing for more general classes of mixed states, a topic that is important for experimental implementations of self-testing, where the ideal of purity can never be fully achieved.
Our approach also shows that certain {\em convexly nonextremal} boxes can self-test a particular entangled state, thereby expanding the scope of possibilities for achieving such a certification.

Finally, our resource-theoretic framing of self-testing also implies that the notion can be extended beyond quantum states and boxes, to arbitrary pairs of resource types,
 and to arbitrary resource theories. Suppose that resources of both types are described in a single resource theory, so that one can specify the partial order over all resources in terms of interconvertibility under the free operations.
Leveraging our abstract characterization of self-testing (given explicitly in Appendix~\ref{equivdefnself}), 
self-testing of a type-$T$ resource, $R$, by a type-$T'$ resource, $R'$, simply means that the upward closure among resources of type $T$ of $R'$  is equal to that of $R$. Equivalently, it means that the upward closure of the equivalence class of $R'$ among equivalence classes of resources of type $T$ has a unique least element, 
 namely, the equivalence class of $R$.
Such a notion of self-testing could be applied, for example, to the resource types
  considered in the type-independent LOSR resource theory laid out in Refs.~\cite{semiquantum,rosset2020characterizing}, which are relevant for semi-device-independent certification protocols. It could also be applied to completely different resource theories in other contexts. 
 
We showed in Section~\ref{Werner} how our resource-theoretic approach naturally accommodates Werner states and hidden nonlocality, in such a way that they are no longer fundamentally puzzling.
It will be interesting to reconsider related results through this lens, such as conversions from multiple entangled states~\cite{BBP+96,Navascues2011,Palazuelos2012} to nonlocal boxes (including entanglement distillation) and Peres' conjecture~\cite{Peres1999,Vertesi2014}. The considerations in Ref.~\cite{semiquantum} on encodings of nonclassicality are likely to be of value here.
 
 It also remains to be seen to what extent the LOSR paradigm can provide
 new insights into the question of what resources are needed for information-processing tasks that are built on Bell inequality violations---such as nonlocal games~\cite{Broadbent2006,Palazuelos2016,Johnston2016}, key distribution~\cite{BHK,Acin2006QKD,Scarani2006QKD,Acin2007QKD,vazirani14,Kaniewski2016chsh} and randomness generation~\cite{Colbeckthesis,colbeckamp,Pironio2010,Dhara2013DIRNG}.

In terms of developing the resource theory of LOSR-entanglement, we have demonstrated the significance of LU-equivalence for the problem of characterizing the necessary and sufficient conditions for convertibility of pure entangled states under LOSR.  Leveraging this result, we have solved the problem completely in the bipartite case, and proven a useful necessary condition for LOSR-convertibility in the multi-partite case.

Our reanalysis of the interplay of entanglement and nonlocality provides, we believe,  a strong motivation for developing the resource theory of LOSR-entanglement to the same level of sophistication as has been achieved for LOCC-entanglement.   Another motivation for doing so is that the resource theory of LOSR-entanglement has  applications beyond the study of states and boxes in Bell scenarios.
    Ref.~\cite{semiquantum} showed that LOSR is the appropriate set of free operations for evaluating the inconvertibility of many other types of resources besides boxes,  including steering assemblages \cite{Einstein1935,Schrodinger1935,wisesteer,Zjawin2023quantifyingepr,Zjawin2023resourcetheoryof} and teleportages \cite{telep, PhysRevA.99.032334,Hoban_2018}, and that doing so
      unifies and extends a variety of seminal results regarding interconversions between these distinct forms of nonclassicality. Consequently, for the study of 
      the interplay of entanglement to 
      these other types of nonclassical resources, it is also LOSR-entanglement that is the appropriate notion.

We note, finally,  that entanglement theory has also recently found applications beyond quantum information processing, such as in many-body physics and in the study of holography.  Given that the notion of classical communication does not seem to be pertinent in any of these applications, there is reason to suspect that LOCC-entanglement might be ill-suited to these applications. 
Our results, therefore, call out for a reassessment of how to formalize the notion of entanglement in these fields of study, 
 and a consideration of whether the paradigm of LOSR-entanglement might be more suitable.\footnote{For instance, Refs.~\cite{Harrow2010,anshu2022entanglement} support the notion that LOSR-entanglement is pertinent to many-body physics, given that it argues for the pertinence of the notion of entanglement spread, which is a 
 pure-state monotone for LOSR-entanglement but not for LOCC-entanglement,
   as noted in Sec.~\ref{sec:RTLOSR}.} 

\section*{Acknowledgments}
We thank Denis Rosset and Jed Kaniewski for valuable discussions, as well as Antonio Ac\'in, Valerio Scarani, and Patrick Hayden for thoughtful feedback on a draft of this article.  
 We thank an anonymous referee for proposing the algorithm described in Appendix H (which is more efficient than the one appearing in an earlier draft of this article), for noticing a mistake in the proof within Appendix \ref{catalysis} (now fixed), and for suggesting various changes that improved the manuscript.
D.S.~is supported by a Vanier Canada Graduate Scholarship. T.C.F.~is supported by the Natural Sciences and Engineering Research Council of Canada (NSERC), grant  411301803.
 R.K.~is supported by the Charg\'e des recherches fellowship of the Fonds de la Recherche Scientifique - FNRS (F.R.S.-FNRS), Belgium. A.B.S.~and D.S.~acknowledge  support by the Foundation for Polish Science (IRAP project, ICTQT, contract no. 2018/MAB/5, co-financed by EU within Smart Growth Operational Programme). This research was supported by Perimeter Institute for Theoretical Physics. Research at Perimeter Institute is supported in part by the Government of Canada through the Department of Innovation, Science and Economic Development Canada and by the Province of Ontario through the Ministry of Colleges and Universities.

\bibliographystyle{quantum}
\bibliography{resourcetheory}

\begin{appendices}
\addtocontents{toc}{\protect\setcounter{tocdepth}{-1}}

\section{Extracting a monotone from each instance of the anomaly} \label{monotones}

As mentioned in the main text, one can define LOSR monotones which capture features of the LOSR preorder that are not captured by LOCC monotones simply by repurposing each of the known examples of ``anomalies of nonlocality''. We now define a few such monotones explicitly. Each example of the anomaly involves a particular nonlocality-witnessing function, that is, a real-valued function over boxes for which there is a threshold value that attests to the box being nonlocal. The natural monotone {\em over quantum states} corresponding to that function is obtained by a yield-construction~\cite{gonda2023monotones}, in which one computes the maximum value of that function over all boxes that can be generated using LOSR operations
from the
 given state. The resulting function over states
  is an LOSR-entanglement monotone, due to the explicit maximization over LOSR operations in the definition.

In the first anomaly we listed below Eq.~\eqref{notconv}, the function which was used to witness the nonlocality of a given box was the maximum probability with which that box could run Hardy's version of Bell's theorem. Denoting this function $f_{\rm prob Hardy}$, one can define an LOSR monotone {\em for states} as follows:
\begin{equation}
M_{\rm probHardy}(\rho) := \max_{\tau \in \LOSR} \{ f_{\rm prob Hardy}(\tau(\rho)) \},
\end{equation}
where $\tau$ is any LOSR operation taking states to boxes.
As an aside, it is
worth noting that the monotone $M_{\rm NPR}$ from Ref.~\cite{wolfe2020quantifying}, which is defined for boxes, is maximized by a particular Hardy box (defined in Table~4 of Ref.~\cite{wolfe2020quantifying}). Hence, another relevant LOSR monotone for states can be defined as follows:
\begin{equation}
M_{\rm NPR}(\rho) := \max_{\tau \in \LOSR} \{ M_{\rm NPR}(\tau(\rho)) \},
\end{equation}
where $\tau$ is any LOSR operation taking states to boxes.

In the second anomaly we listed below Eq.~\eqref{notconv}, the function which was used to witness nonlocality was a tilted Bell inequality (viewed as a function from boxes to the reals). Denoting the function defined by the tilted Bell inequality with tilt $\alpha$~\cite{Yang2013selftesting,Bamps2015selftesting} as $f_{\alpha-{\rm tilted}}$, an LOSR monotone for states can be defined as follows:
\begin{equation}
M_{\alpha-{\rm tilted}}(\rho) := \max_{\tau \in \LOSR} \{ f_{\alpha-{\rm tilted}}(\tau(\rho)) \},
\end{equation}
where $\tau$ is any LOSR operation taking states to boxes.

By now, the pattern is clear. One can define an analogous yield-based monotone for each of the anomalies.

We leave the task of using these monotones to glean insights into the LOSR preorder over quantum states for future work. A deficiency of the above monotones (as presented) is that it may be difficult to perform the optimization in their definitions (especially since the convex set which is optimized over does not have a finite set of extreme points). Those monotones which are defined in terms of an abstract function (such as $f_{\rm KR}$) will also be at least as difficult to compute as the functions themselves. In most cases, finding a closed form expression for a given monotone is paramount for it to be practically useful.

\section{Proving LOSR incomparability between two Bell-pairs and the GHZ state} \label{BellGHZ}

The purpose of this section is to prove Eq.~\eqref{noconvBellGHZ}, i.e. the LOSR incomparability of the tripartite pure states
\begin{align}
    \ket{\psi_{\rm 2Bell}} &\equiv |\phi^{+}\rangle_{A_1B}\otimes  |\phi^{+}\rangle_{A_2C}, \text{ and }\\
    |\psi_{\rm GHZ}\rangle &\equiv \tfrac{1}{\sqrt{2}} (|000\rangle_{ABC} +|111\rangle_{ABC}),
\end{align}
as was required for resolving the anomaly regarding genuinely $3$-way entanglement.

To prove incomparability, we must show that neither state can be converted into the other. The impossibility of each of these two conversions is proved by showing that the necessary condition given in Eq.~\eqref{eq:multipartite_extension} is not satisfied.
For tripartite systems, there are exactly three bipartitions among the three parties, henceforth denoted $\{A|BC, B|AC, C|AB\}$. The associated squared Schmidt coefficients for $\ket{\psi_{\rm 2Bell}}$ and $|\psi_{\rm GHZ}\rangle$ are as follows:
\begin{equation}
    \begin{alignedat}{3}
        \lambda^{(A|BC)}_{\rm 2Bell}
        &= (\tfrac{1}{4},\tfrac{1}{4},\tfrac{1}{4},\tfrac{1}{4}),
        \qquad 
        &&\lambda^{(A|BC)}_{\rm GHZ}
        &&= (\tfrac{1}{2},\tfrac{1}{2}),\\
        \lambda^{(B|AC)}_{\rm 2Bell}
        &= (\tfrac{1}{2},\tfrac{1}{2}),
        &&\lambda^{(B|AC)}_{\rm GHZ}
        &&= (\tfrac{1}{2},\tfrac{1}{2}),\\
        \lambda^{(C|AB)}_{\rm 2Bell}
        &= (\tfrac{1}{2},\tfrac{1}{2}),
        &&\lambda^{(C|AB)}_{\rm GHZ}
        &&= (\tfrac{1}{2},\tfrac{1}{2}).
    \end{alignedat}
\end{equation}

First, when considering the $A|BC$ partition, it becomes clear that $|\psi_{\rm GHZ}\rangle \not \mapsto |\psi_{\rm 2Bell}\rangle$ under LOSR as there is no vector $\lambda^{(A|BC)}_{\zeta}$ such that $(\lambda^{(A|BC)}_{\rm GHZ})^{\downarrow} = (\lambda^{(A|BC)}_{\rm 2Bell}\otimes \lambda^{(A|BC)}_{\zeta})^{\downarrow}$, and thus Eq.~\eqref{eq:multipartite_extension} cannot be satisfied. (One can also see this through the failure of the condition in Corollary~\ref{rankcoro} in Appendix~\ref{alg}, since the ratio of the Schmidt rank of $|\psi_{\rm GHZ}\rangle$ to that of $|\psi_{\rm 2Bell}\rangle$
 is not an integer.)
Second, Eq.~\eqref{eq:multipartite_extension} implies that $|\psi_{\rm 2Bell}\rangle \mapsto |\psi_{\rm GHZ}\rangle$ only if there exists an auxiliary state $\ket{\zeta}$ such that
\begin{align}
    \lambda^{(A|BC)}_{\zeta} &= (\tfrac{1}{2},\tfrac{1}{2}),\label{kk1}\\
    \lambda^{(B|AC)}_{\zeta} &= (1),\label{kk2}\\
    \lambda^{(C|AB)}_{\zeta} &= (1).\label{kk3}
\end{align}
However, these squared Schmidt coefficients are not consistent with any tripartite state $\ket \zeta$. This can be seen by the fact that Eq.~\eqref{kk2} implies that $\ket \zeta$ factorizes across the $B|AC$ partition, and Eq.~\eqref{kk3} implies that $\ket \zeta$ factorizes across the $C|AB$ partition, and together these two facts imply that $\ket \zeta$ must factorize across the $A|BC$ partition. But this contradicts Eq.~\eqref{kk1}, since the latter can only hold if $\ket \zeta$ is entangled
  across the $A|BC$ partition. Therefore, $|\psi_{\rm 2Bell}\rangle \not \mapsto |\psi_{\rm GHZ}\rangle$ under LOSR.

\section{A subtlety regarding the standard definition of self-testing of states}\label{Oversight}

In this appendix, we discuss a subtlety regarding our claim that Definition~\ref{defnselftesteOld} is the standard definition of self-testing of states.  
Careful inspection will reveal that Definition~\ref{defnselftesteOld} is, in fact, not equivalent to the definition one finds in the primary review article on the topic~\cite{vsupic2020self}.  The criterion for a state $|\psi\rangle$ to be self-tested by a box $B$ stipulated therein  (their Definition 1) is similar to Definition~\ref{defnselftesteOld},
  but without the condition that \enquote{$B$ can be obtained from $|\psi\rangle$ by local measurements.} Nonetheless, we argue that this omission 
 was surely inadvertent.  Definition~\ref{defnselftesteOld} is the notion of self-testing of states that follows from the original notion of self-testing, where the sort of object for which the notion of self-testability applied was a triple, consisting of a state on $AB$, a set of local projective measurements on $A$, and a set of local projective measurements on $B$.  
 The original definition, introduced by Mayers and Yao~\cite{mayers1998quantum,mayers2003self} and still used today~\cite{Yang2013selftesting,Wang2016,Bancal2015SelfTesting,Supi2018Mutliselftest} is essentially as follows:

\begin{defn}[original]
\label{defnselftesteOriginal}
Consider a triple  consisting of a pure state $|\psi\rangle_{AB}$, a set of projective measurements $\{ \mathcal{M}^A_s\}_s$ (where $\mathcal{M}^A_s$ is represented by a projector-valued measure $\{M^{A}_{x|s}\}_x$), and a set of projective measurements $\{ \mathcal{M}^B_t\}_t$ (where $\mathcal{M}^B_t$ is represented by a projector-valued measure $\{M^{B}_{y|t}\}_y$). 
 The triple $(|\psi\rangle_{AB}, \{ \mathcal{M}^A_s\}_s, \{\mathcal{M}^B_t\}_t)$ is self-tested by a  box $B$  if for any triple $(|\phi\rangle_{A'B'}, \{ \mathcal{N}^{A'}_s\}_s, \{ \mathcal{N}^{B'}_t\}_t)$ (consisting of a pure state and two sets of projector-valued measures) from which the probability distribution $P_{XY|ST}$ associated to $B$ can be obtained as $P_{XY|ST}(xy|st)= {\rm Tr}_{AB}\left(|\phi\rangle \langle \psi|_{A'B'} (N^{A'}_{x|s} \otimes N^{B'}_{y|t}) \right)$, there is a local isometry  that takes $N^{A'}_{x|s} \otimes N^{B'}_{y|t}|\phi\rangle_{A'B'}$ to $\big( M^{A}_{x|s} \otimes M^{B}_{y|t}|\psi\rangle_{AB} \big) \otimes |\xi\rangle$ for some $|\xi\rangle$.
 \end{defn}

Subsequently, this definition was dissected into a definition of self-testability of states and a definition of self-testability of measurements~\cite{ScaraniStateSelfTest2009,JedWeak2020,Coladangelo2017,vsupic2020self}.  However, in the process, something was inadvertently left  on the operating table.  

Consider what Definition~\ref{defnselftesteOriginal} implies for the conditions under which a box $B$ self-tests a given state $|\psi\rangle$, as opposed to self-testing a state together with a pair of measurements. 

First, note that Definition~\ref{defnselftesteOriginal} implies that the probability distribution $P_{XY|ST}$ associated to the box $B$ can be obtained from the triple $(|\psi\rangle_{AB}, \{ \mathcal{M}^A_s\}_s, \{\mathcal{M}^B_t\}_t)$ via $P_{XY|ST}(xy|st)= {\rm Tr}_{AB}\left(|\psi\rangle \langle \psi|_{AB} (M^{A}_{x|s} \otimes M^{B}_{y|t}) \right)$.  
This follows from the fact that by assumption $P_{XY|ST}(xy|st)= {\rm Tr}_{AB}\left(|\phi\rangle \langle \phi|_{A'B'} (N^{A'}_{x|s} \otimes N^{B'}_{y|t}) \right)$ and by replacing $N^{A'}_{x|s} \otimes N^{B'}_{y|t}|\phi\rangle_{A'B'}$ and its adjoint by their images under the isometry, namely, $\big( M^{A}_{x|s} \otimes M^{B}_{y|t} |\psi\rangle_{AB} \big) \otimes |\xi\rangle$ and its adjoint.
Consequently, it follows from the definition that the box $B$ must be obtainable from $|\psi\rangle$ by local projective measurements. 
 
Second, note that 
 the isometry that takes $N^{A'}_{x|s} \otimes N^{B'}_{y|t}|\phi\rangle_{A'B'}$ to $\big( M^{A}_{x|s} \otimes M^{B}_{y|t}|\psi\rangle_{AB} \big) \otimes |\xi\rangle$ in Definition~\ref{defnselftesteOriginal} also takes $|\phi\rangle_{A'B'}$ to $|\psi\rangle_{AB} \otimes |\xi\rangle$.  This follows from summing over $x$ and $y$ and recalling that by virtue of the measurements being projective, we have $\sum_x  N^{A'}_{x|s} =I^{A'}$, $\sum_y  N^{B'}_{y|s} =I^{B'}$, $\sum_x  M^{A}_{x|s} =I^{A}$, and $\sum_y  M^{B}_{y|s} =I^{B}$.  Consequently, for any state $|\phi\rangle$ that can generate $B$ by local projective measurements, there is a local isometry that takes $|\phi\rangle$ to $|\psi\rangle \otimes |\xi \rangle$ for some $|\xi\rangle$.

  Definition~\ref{defnselftesteOriginal} therefore
  implies that a state is self-tested by a box if and only if the following two conditions are both satisfied:
\begin{quote} Box-generating condition: The state $|\psi\rangle$ must be able to generate the box $B$ by local projective measurements.
\end{quote}
\begin{quote} Extraction condition: The state can be extracted from any state that could have generated the box by implementing a local isometry and tracing over an auxiliary system which is in a pure state.
\end{quote}
These two conditions are incorporated into what we have termed the ``standard definition'' (Definition~\ref{defnselftesteOld})---we have merely relaxed the condition that the measurements be projective and that the auxiliary system be in a pure state.  But the definition one finds in Ref.\cite{vsupic2020self}
 has left out the box-generating condition.

 It consequently 
asserts that  $|\psi\rangle$ can be self-tested by box $B$ even if $|\psi\rangle$ itself could not have generated $B$.  In particular, therefore, it asserts that {\em every nonlocal box $B$ self-tests every 
product pure state}, simply because a product
 pure state can be extracted from {\em any} pure state by 
\rob implementing a local isometry and tracing over an auxiliary system which is in a pure state\blk.
 This conclusion is one that the authors presumably did not intend to endorse.  Indeed, it is commonplace to refer to self-testing as a certification of what state {\em was used} to generate the box.  This example also highlights the inconsistency of the definition one finds in Ref.~\cite{vsupic2020self} with the original definition (Definition~\ref{defnselftesteOriginal}).  According to the latter, if $|\psi\rangle$ is a product state, then it {\em cannot}  be self-tested by a nonlocal box because  the local measurements are explicitly required to be on the state $|\psi\rangle$ rather than on the state $|\xi\rangle$ of the auxiliary system, and a product state cannot yield a nonlocal box. 

It is these considerations that suggest to us that the omission of the box-generating condition was merely an oversight, and that the standard view of self-testing takes it as an implicit additional condition. Our Definition~\ref{defnselftesteOld} reflects this by explicitly including the condition.

\section{General resource-theoretic definitions of self-testing} \label{equivdefnself}

The definition of self-testing that is appropriate for a generic resource theory and two types of resources is as follows:
\begin{defn}[general resource-theoretic]\label{defnselftesteGeneral}
{
\abovedisplayskip=1pt plus 3pt
\belowdisplayskip=1pt plus 3pt
\abovedisplayshortskip=1pt plus 3pt
\belowdisplayshortskip=1pt plus 3pt
Consider a resource theory with the free operations $\mathcal{F}$, a set of resources that are the testees, $\mathcal{E}$,
 and a set of resources that are the testers, $\mathcal{T}$
  (typically the testee and tester resources are presumed to be of different types).  Given a pair consisting of a testee resource $E \in \mathcal{E}$ and a tester resource $T \in\mathcal{T}$, we say that $E$  {\em is self-tested by} $T$ 
   if it holds that
\begin{align}
E \mapsto T \nonumber
\end{align}
 and that for all testee resources $E' \in \mathcal{E}$,
\begin{align} 
\textrm{if }\; E' \mapsto T \; \textrm{ then }\; E' \mapsto E,\nonumber
\end{align}
where all conversions are evaluated relative to the free operations $\mathcal{F}$. A testee resource $E \in \mathcal{E}$ is said to be {\em self-testable} 
 if there exists some $T \in \mathcal{T}$ that self-tests $E$.
}
\end{defn}

For the resource theory of interest in this article, and for self-testing of states by boxes (Definition~\ref{defnselfteste}), the set of testee resources is the set of all states, while
the set of tester resources is the set of all boxes,
 and the set of free operations is LOSR.

We now elaborate on a more abstract resource-theoretic characterization of self-testing which is equivalent to Definition~\ref{defnselftesteGeneral}.

The free operations of the resource theory, $\mathcal{F}$, define a preorder over all of the testee and tester resources of interest, that is,  $\mathcal{E}\cup \mathcal{T}$.   Specifically, for $E\in \mathcal{E}$ and $T \in \mathcal{T}$, we have a preorder relation $E > T$ if there exists an operation in $\mathcal{F}$ converting $E$ into $T$. 

The \emph{upward closure} of a resource $R$ in the preorder of resources defined by a resource theory with free operations $\mathcal{F}$, denoted here by ${\rm UC}_{\mathcal{F}}(R)$, is the set of all resources that can be converted to $R$ using operations in $\mathcal{F}$. Formally, 
\begin{align*}
{\rm UC}_{\mathcal{F}}(R) := \{ R' :R' \mapsto R\},
\end{align*}
where $R' \mapsto R$ here denotes convertibility relative to $\mathcal{F}$. 
Note that it is defined for any $R$, including any testee resource $E\in \mathcal{E}$  and any tester resource $T\in \mathcal{T}$.

Using the notion of upward closure, one can provide an equivalent formulation of Definition~\ref{defnselftesteGeneral}.  
\begin{defn}[general resource-theoretic] \label{defnselftesteGeneral2}
Consider a set of testee resources $\mathcal{E}$ and a set of tester resources $\mathcal{T}$ within a resource theory with free operations $\mathcal{F}$.  A particular testee resource $E \in \mathcal{E}$  is {\em self-tested} by a particular tester resource $T\in \mathcal{T}$ 
 if it holds that the upward closure of $T$ within $\mathcal{E}$ contains all and only those testees which are contained in the upward closure of $E$ within $\mathcal{E}$,
\beq\label{eqgeneral2}
   {\rm UC}_{\mathcal{F}}(T) \cap \mathcal{E} = {\rm UC}_{\mathcal{F}}(E) \cap \mathcal{E}.
   \eeq
\end{defn}

It is also instructive to consider a reformulation of the definition of self-testing which focuses on the partial order of equivalence classes of resources rather than the preorder of resources themselves.  Let $\tilde{R}$ denote the $\mathcal{F}$-equivalence class of the resource $R$.   For a pair of distinct equivalence classes, $\tilde{R}$ and $\widetilde{R}'$, let $\tilde{R}>\widetilde{R}'$ 
denote that 
 $\tilde{R}$ is above 
 $\widetilde{R}'$ in the partial order, meaning that $\forall R\in \tilde{R}, \forall{R'}\in\widetilde{R'}: R \mapsto R'$.

Clearly, the upward closure of a resource $R$ is the same as that of any resource within the $\mathcal{F}$-equivalence class of $R$, that is, ${\rm UC}_{\mathcal{F}}(R)$  is the same for any $R \in \tilde{R}$.  

The \emph{upward closure in the partial order} of an equivalence class of resources, $\tilde{R}$, denoted here by $\widetilde{{\rm UC}}_{\mathcal{F}}(\tilde{R})$, is the set of all equivalence classes of resources that are above or equal to $\tilde{R}$ in the partial order induced by the free operations $\mathcal{F}$. Formally, 
\begin{align*}
\widetilde{{\rm UC}}(\tilde{R}) := \{ \tilde{R}' :\tilde{R}' \ge \tilde{R}\}.
\end{align*}

Using this notion, we can provide another definition of self-testing.

\begin{defn}[general resource-theoretic, equivalence classes] \label{GeneralDefnSelfTestingAbstract}
Consider a set of equivalence classes of testee resources $\tilde{\mathcal{E}}$ and a set of equivalence classes of tester resources $\tilde{\mathcal{T}}$ within a resource theory with free operations $\mathcal{F}$.  A particular testee equivalence class $\tilde{E} \in \tilde{\mathcal{E}}$  is {\em self-tested} by a particular tester equivalence class $\tilde{T}\in \tilde{\mathcal{T}}$ 

 if it holds that the upward closure of $\tilde{T}$ within $\tilde{\mathcal{E}}$ contains all and only those testees which are contained in the upward closure within $\tilde{\mathcal{E}}$ of $\tilde{E}$ ,
\beq\label{abstracteq1}
   \widetilde{\rm UC}_{\mathcal{F}}(\tilde{T}) \cap \tilde{\mathcal{E}} = \widetilde{\rm UC}_{\mathcal{F}}(\tilde{E}) \cap \tilde{\mathcal{E}}. 
   \eeq
This condition is equivalent to $\tilde{E}$ being the minimum element of the upward closure within $\tilde{\mathcal{E}}$ of $\tilde{\mathcal{T}}$ (considered as a subset of a partial order),

\begin{align}
\tilde{E}= 
\min
 \left[ \widetilde{{\rm UC}}_{\mathcal{F}}(\tilde{T}) \cap \tilde{\mathcal{E}}
 \right],
\end{align}
where $\min (S)$ denotes the minimum element of the subset $S$ of a partially ordered set, i.e., an element of $S$ that is smaller than every other element of $S$.
\end{defn}

If one particularizes this definition to the case of states and boxes, one obtains the following definition, which is equivalent to Definition~\ref{defnselfteste}.

\begin{defn} \label{AbstractDefnSelfTesting2}
Consider the resource theory wherein the free operations are LOSR. We say that a particular equivalence class of states $\tilde{\rho}$  is {\em self-tested}  by a particular equivalence class of boxes $\tilde{B}$ if $\tilde{\rho}$ is the minimum element of the upward closure  of $\tilde{B}$ within the set of all equivalence classes of states, $\widetilde{\rm States}$, that is, if
\begin{align}\label{eqrhoUC}
\tilde{\rho} = 
\min \left[ \widetilde{{\rm UC}}(\tilde{B}) \cap \widetilde{\rm States}
 \right].
\end{align}
\end{defn}

The resource-theoretic perspective also motivates a natural relaxation of the notion of self-testing, corresponding to dropping the condition of uniqueness.  To express the relaxation of Definition~\ref{GeneralDefnSelfTestingAbstract}, for instance, in which a {\em set} of equivalence classes of testee resources, 
$\{\tilde{E}_x \}_x \subset \tilde{\mathcal{E}}$
 (which could be finite or continuous) is self-tested by an equivalence class of tester resources, $\tilde{T}$, it suffices to replace Eq.~\eqref{abstracteq1} with
\beq\label{abstracteq1relaxed}
   \widetilde{\rm UC}_{\mathcal{F}}(\tilde{T}) \cap \tilde{\mathcal{E}} = \widetilde{\rm UC}_{\mathcal{F}}(\{\tilde{E}_x \}_x) \cap \tilde{\mathcal{E}},
   \eeq
   which is equivalent to
   \begin{align}
\{\tilde{E}_x \}_x = 
 {\rm MinElements}
\left[ \widetilde{{\rm UC}}_{\mathcal{F}}(\tilde{T}) \cap \tilde{\mathcal{E}}
 \right],
\end{align}
where ${\rm MinElements}(S)$ denotes the minimal elements of a subset $S$ of a partially ordered set, i.e., those elements of $S$ that are not greater than any other element of $S$.  
\rob
(If a subset $S$ has a minimum element and ${\rm MinElements}(S)$ is a singleton set, then ${\rm MinElements}(S)$ coincides with ${\rm min}(S)$.\footnote{ To see that ${\rm MinElements}(S)$ being a singleton set does not by itself imply that it coincides with ${\rm min}(S)$, it suffices to consider preorders that can be of infinite height and therefore that can include subsets with no minimum. For instance, if $S$ is an infinite chain together with another element $r$ that is incomparable to the points in the chain, then ${\rm MinElements}(S) =\{r\}$, but ${\rm min}(S)$ is the empty set. We thank an anonymous referee for this example. }  )

To recast this relaxed notion of self-testing in terms of a set of resources 
$\{ E_y \}_y
 \subset \mathcal{E}$ rather than a set of equivalences classes thereof, 
 i.e., to express this relaxation in the form of Definition~\ref{defnselftesteGeneral2}, it suffices to modify Eq.~\eqref{eqgeneral2} to:
\beq\label{eqgeneral2relaxed}
   {\rm UC}_{\mathcal{F}}(T) \cap \mathcal{E} = {\rm UC}_{\mathcal{F}}(\{ E_y \}_y ) \cap \mathcal{E}.
\eeq
Note that if the set of equivalence classes of testee resources $\{ \tilde{E}_x \}_x \subset \tilde{\mathcal{E}}$ is self-tested by the box $B$, then the only sets of testee resources, 
$\{ E_y\}_y
\subset \mathcal{E}$, that are self-tested by $B$ are sets containing at least one representative from {\em each} equivalence class in $\{\tilde{E}_x \}_x$.

The notion of {\em unique} self-testing differs from this relaxed notion of self-testing by the fact that it demands that the set of equivalence classes of testee resources, $\{ \tilde{E}_x\}_x$, is a singleton set, or equivalently, that the set of testee resources, $\{ E_y\}_y$, is contained within a single equivalence class.

To be sure, in cases of nonuniqueness, there will be more uncertainty about the identity of the testee resource than if unique certification were possible.  But for certain applications, it might be sufficient if this uncertainty is quantified and bounded.

We now consider the specialization of this relaxed notion of self-testing to the case of self-testing entangled states by nonlocal boxes.  The condition for a set of equivalence classes of states $\{\tilde{\rho}_x \}_x$ to be self-tested by an equivalence class of boxes, $\tilde{B}$, is as in Definition~\ref{AbstractDefnSelfTesting2}, but where Eq.~\eqref{eqrhoUC} is replaced by:
\begin{align}\label{eqrhoUC2}
\{ \tilde{\rho}_x \}_x = 
\rm MinElements
\left[ \widetilde{{\rm UC}}(\tilde{B}) \cap \widetilde{\rm States}
 \right].
\end{align}

In the main text, we noted that a chiral state is an example of a state that cannot be self-tested according to the nonrelaxed definition.  From the perspective of the relaxed definition, however, 
the set consisting of the equivalence class of a chiral state and the equivalence class of its image under complex conjugation is one that {\em can} be self-tested. 
Note that allowing that a box can certify a {\em set of equivalence classes of states} rather than a single such equivalence class (up to upward closure of course) is quite  different from proposing that the free operations be modified to include complex conjugation.
As we argued in the main text, such a relaxation is the best way to interpret proposals for modifying the notion of self-testing in order to include chiral states.

\section{Self-testing as unique certification up to equivalence}\label{selftestinguptoequivalence}

We noted in Section~\ref{sec:selftesting} that the resource-theoretic perspective on self-testing clarifies that it corresponds to certification up to {\em upward closure}, rather than up to equivalence. 

Nonetheless, there is a circumstance in which one {\em can} certify the state up to equivalence, namely, if one has restricted  attention to a subset of states wherein no two states are strictly ordered.   Under such a promise about the state, one can conclude from the satisfaction of the condition of self-testability of $\rho$ by box $B$ that any other state $\sigma$ in the subset that can generate the box $B$ is indeed LOSR-equivalent to $\rho$, rather than merely above $\rho$ in the LOSR order.

Determining when {\em certification up to equivalence} is possible motivates the task of finding  subsets of states that contain no strictly ordered elements, so that every pair of states in the set is either equivalent or incomparable. 
 That is precisely the sort of problem that is solved by determining the preorder of entangled states.  
In particular, our Corollary~\ref{cor:equivalent_or_incomparable} establishes that for any set of multipartite pure states having a fixed Schmidt rank along every bipartition, every pair of states in the set is either LOSR-equivalent or LOSR-incomparable.  As such, for the set of all {\em bipartite} pure states with a given Schmidt rank, every pair of states is either LOSR-equivalent or LOSR-incomparable.  If one then leverages the result of Ref.~\cite{Coladangelo2017}, namely, that every pure bipartite state is self-testable, it follows that if the state is known to be in the set of  pure bipartite states of a given Schmidt rank, then it {\em is} possible to certify the state up to equivalence.

We note that the existence of a continuum of states that can be uniquely certified in the sense just described is only possible because of two features of the LOSR preorder: it contains anti-chains (sets of pairwise incomparable elements) of infinite cardinality both among states, as follows from Corollary~\ref{cor:equivalent_or_incomparable} and the continuum of possibilities for vectors of squared Schmidt coefficients of a given Schmidt rank, and among boxes, as proven in Ref.~\cite{wolfe2020quantifying}.   

For what other sets of pure states might one be able to prove the possibility of unique certification of states within that set up to LOSR-equivalence?  Our characterization of the LOSR preorder of pure states also sheds light on this question.   
As we demonstrate in Corollary~\ref{rankcoro} of Appendix~\ref{alg}, 
a sufficient condition for the incomparability of two LOSR-inequivalent bipartite pure states is that the ratio of their Schmidt ranks is not an integer.
We can then leverage the result of Ref.~\cite{Coladangelo2017} to conclude the following:
 if we let $S_x$ be the set of all pure bipartite states of Schmidt rank $x$, then for any set of integers $\mathcal{X}$ for which every ratio is noninteger
  (i.e., $\forall x,x'\in \mathcal{X}$ it holds that $x/x' \not\in \mathbb{Z}^{+}$), the set of states $\cup_{x \in \mathcal{X}} S_x$  
is one relative to which unique certification of every state  up to LOSR equivalence) is achievable,
 while for any set of integers $\tilde{\mathcal{X}}$ for which there exists a pair with an integer ratio 
 (i.e., $\exists x,x'\in \tilde{\mathcal{X}}$ such that $x/x' \in \mathbb{Z}^{+}$), such unique certification is impossible. \rob Presumably, this example is of academic interest only.  \blk

\section{On LOSR equivalence classes and LO equivalence classes of states and boxes}
\label{discrepanciesLOSRLO}

We begin by exhibiting  a collection of mixed states on a bipartite system $A'B'$ that are LOSR-equivalent to a pure state $|\psi\rangle \langle \psi|$ on the bipartite system $AB$.\footnote{Note that we denote the quantum system in Bob's possession by $B$, a notation that we have also used for a generic box; which usage is intended in a given instance should be clear from the context. } Specifically, consider the states of the following form, which we term {\em flag-convexifications} of $\psi$:
\begin{align}\label{flaggedstate}
\rho_{A'B'} := \sum_{ij} p(i j) U_A^{(i)} \otimes U_B^{(j)} & \ket{\psi} \! \bra{\psi}_{AB} U_A^{ (i)\dagger} \otimes U_B^{ (j)\dagger} 
\nonumber \\ &\otimes \ket{i}\! \bra{i}_{Z_A} \otimes \ket{j}\! \bra{j}_{Z_B},
\end{align}
where $Z_A$ and $Z_B$ are local flag degrees of freedom for Alice and for Bob respectively, $A'=(A,Z_A)$ and $B'=(B,Z_B)$, $p(ij)$ is a joint probability distribution, $U_A^{(i)}$ and $U_B^{(j)}$ are arbitrary unitaries, and $\{ \ket{i}_{Z_A}\}_i$ and $\{ \ket{j}_{Z_B}\}_j$ are arbitrary orthonormal bases for $Z_A$ and $Z_B$.     As long as $p(ij)$ is not a point distribution, 
$\rho_{A'B'}$ is mixed.  Such a state is depicted schematically in Fig.~\ref{NonextremalStates}.
Note that this set of states is not an exhaustive characterization of those that are LOSR-equivalent to $\ket{\psi}$.\footnote{The latter set can be obtained by considering all the ways of supplementing $\ket{\psi}$ with shared randomness and then applying local isometries.}

To give the intuition for why a flag-convexification of $\ket{\psi}$, that is, a mixed state $\rho$ of the form of Eq.~\eqref{flaggedstate}, is LOSR-equivalent to $\ket{\psi}$, we provide a schematic in Fig.~\ref{NonextremalStates} of the reversible LOSR operation that is needed to take $|\psi\rangle$ to $\rho$; the one that recovers $|\psi\rangle$ from $\rho$ is easily inferred from the figure.

\begin{figure}[h!]
\centering
    \includegraphics[scale=1]{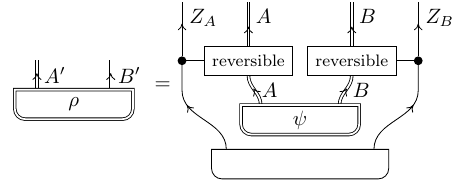}
    \caption{A schematic depiction of a mixed state $\rho$ of the form of Eq.~\eqref{flaggedstate}, which is LOSR-equivalent to the pure state $|\psi\rangle$.
We have here used the conventional notation for a controlled operation.
    }
    \label{NonextremalStates}
\end{figure}

The detailed proof 
 is as follows. 
  To convert $\ket{\psi}_{AB}$ to $\rho_{AB}$, the parties use their shared randomness to prepare their flag systems in the separable state $\sum_{ij} p(ij) \ket{i}\bra{i}_{Z_A} \otimes \ket{j}\bra{j}_{Z_B}$, Alice applies a controlled unitary $\sum_i U_A^{(i)} \otimes \ket{i}\bra{i}_{Z_A}$ with her local flag system $Z_A$ as the control and $A$ 
  as the target, and Bob acts similarly.
To convert $\rho_{AB}$ to $\ket{\psi}_{AB}$, Alice applies the controlled unitary $\sum_i U_A^{(i)\dagger} \otimes \ket{i}\bra{i}_{Z_A}$, 
 Bob acts similarly, and then each traces out their local flag system.

Flag-convexification also provides a means of exhibiting a collection of mixed states that are LO-equivalent (as well as LOSR-equivalent) to a given pure state.  It suffices to take the distribution $p(ij)$ in Eq.~\eqref{flaggedstate} to be factorizing, i.e., $p(ij)= p(i)p(j)$.  In this case, one only requires local randomness, rather than shared randomness to prepare $p(ij)$, so that the operation in Eq.~\eqref{flaggedstate} is implementable by LO.

It is also evident that mixed states of the form of Eq.~\eqref{flaggedstate} wherein $p(ij)$ does {\em not} factorize are examples of states that are equivalent to $\ket{\psi}$ in the LOSR order, but strictly above $\ket{\psi}$ in the LO order.   The existence of such states was noted in the main text.  This was claim (b) of Lemma~\eqref{lem:conversion3}.  The example used to prove the claim was $\rho = |\psi\rangle\langle \psi| \otimes \omega$, where $\omega$ is any separable but nonfactorizing mixed state.  Flag-convexification merely provides a collection of further examples that establish this claim.   Note that the example we used to prove claim (b) of Lemma~\ref{lem:conversion3} in the text is a trivial instance of Eq.~\eqref{flaggedstate} with nonfactorizing $p(ij)$, namely, one wherein the reversible operations are identity maps.

We now turn from states to boxes.  The arguments proceed analogously to the case of states.

We begin by exhibiting a collection of convexly nonextremal boxes that are LOSR-equivalent to a convexly extremal box $B_{\rm ext}$, termed {\em flag-convexifications} of $B_{\rm ext}$.
Denoting the conditional probability distribution associated to the convexly extremal box by $B^{\rm ext}_{X Y|ST}$ (where $X$ and $Y$ are the outputs for Alice and Bob and $S$ and $T$ are their inputs), one can construct from $B_{\rm ext}$ a box $B$ associated to the conditional probability distribution $B_{X' Y'|S T}$, as follows:
\begin{align}\label{flaggedbox}
B_{X'Y'|ST} = 
\sum_{ij} p(ij) \left[ \mathcal{F}_X^{(i)} \otimes \mathcal{F}_Y^{(j)} \right] ( B^{\rm ext}_{XY|ST} )   \otimes \delta_{Z_A,i} \delta_{Z_B,j}
\end{align}
where $X' =(X,Z_A)$ and $Y'=(Y,Z_B)$, where $\mathcal{F}_X^{(i)}$ is a reversible function acting on $X$ and similarly for $\mathcal{F}_Y^{(j)}$, and where $\delta_{x,y}$ represents the Kronecker-delta function.  Such a box is depicted schematically in Fig.~\ref{NonextremalBoxes}.
As long as $p(ij)$ is not a point distribution, $B$ is convexly nonextremal.\footnote{Note that this is not the most general form of a convexly nonextremal box that is LOSR-equivalent to $B_{\rm ext}$.  In particular, by copying the inputs of $B_{\rm ext}$, i.e., $S$ and $T$, and feeding these forward, the reversible function applied to $X$ can be controlled on $S$ in addition to $Z_A$ and the reversible function applied to $Y$ can be controlled on $T$ in addition to $Z_B$.}

The intuition for why any $B$ that is a flag-convexification of $B_{\rm ext}$ is LOSR-equivalent to  $B_{\rm ext}$, is provided in Fig.~\ref{NonextremalStates} which exhibits the reversible LOSR operation that is needed to take $B_{\rm ext}$ to $B$ and makes clear what LOSR operation would recover $B_{\rm ext}$ from $B$.

\begin{figure}[h!]
\centering
    \includegraphics[scale=1]{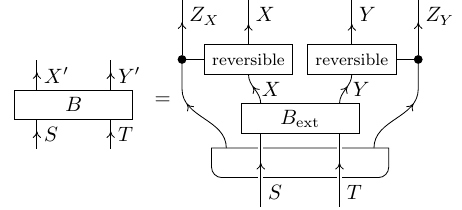}
    \caption{
    A schematic depiction of a convexly nonextremal box $B$ of the form of Eq.~\eqref{flaggedbox}, which is LOSR-equivalent to the convexly extremal box $B_{\rm ext}$.
}
    \label{NonextremalBoxes}
\end{figure}

The detailed proof 
 is as follows.   To convert $B_{\rm ext}$ to $B$, the parties use their shared randomness to prepare their classical flag variables $Z_A$ and $Z_B$ in the distribution ${\sum_{ij} p(ij) \delta_{Z_A,i} \otimes \delta_{Z_B,j}}$, Alice applies the function $\mathcal{F}^{(i)}_{X}$ to $X$ by controlling on the value of $Z_A$, and Bob applies the function $\mathcal{F}^{(i)}_{Y}$ to $Y$ by controlling on the value of $Z_B$. To convert $B$ to $B_{\rm ext}$, Alice simply also applies the same controlled operation, as does Bob, and then each traces out their local flag variable.

As with states, flag-convexification of boxes provides a means of exhibiting a collection of mixed states that are LO-equivalent, rather than LOSR-equivalent, to a given convexly extremal box, namely, those wherein $p(ij)$ in Eq.~\eqref{flaggedbox} factorizes.  Also similarly to the case of states, boxes of the form of Eq.~\eqref{flaggedbox} wherein $p(ij)$ does {\em not} factorize are equivalent to $B_{\rm ext}$ in the LOSR order but above $B_{\rm ext}$ in the LO order.  This is claim (b) of Lemma~\ref{convext}.  The example used to prove this claim in the text was a trivial instance of a flag-convexified box with nonfactorizing $p(ij)$, namely, one wherein the reversible operations are identity maps.

A graphical depiction of the difference between the partial order of states and boxes under LOSR and the partial order under LO is provided in Fig.~\ref{LOvsLOSRorders}.
 
 \begin{figure}[h!]
 \centering
    \subfloat[]{
        \raisebox{-0.5\height}{\includegraphics[scale=1]{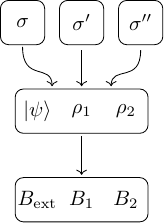}}
        \label{fig:hasse_losr}
    }
    \hspace{3mm}
    \subfloat[]{
        \raisebox{-0.5\height}{\includegraphics[scale=1]{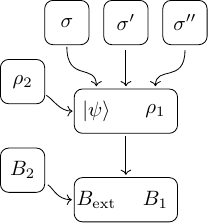}}
        \label{fig:hasse_lo}
    }
    \caption{Partial orders of equivalence classes of states and boxes relative to LOSR (a) and relative to LO (b) revealing the discrepancies between LOSR-based and LO-based notions of self-testing.  Here, the $\rho$'s and the $B$'s represent particular elements of a given equivalence class, where a class is depicted by a loop. $\rho_1$ and $\rho_2$ are of the form of Eq.~\eqref{flaggedstate} and $B_1$ and $B_2$ are of the form of Eq.~\eqref{flaggedbox}, but only $\rho_2$ and $B_2$ involve a distribution $p(ij)$ that does not factorize.  
 $\rho_2$ is self-tested by $B_{\rm ext}$ relative to LOSR but not relative to LO, and $B_2$ self-tests $|\psi\rangle$ relative to LOSR but not relative to LO.
   }
    \label{LOvsLOSRorders}
\end{figure}

\section{On the insufficiency of the condition in Corollary~\ref{cor:multipartite_extension} for LOSR-convertibility} \label{insufficiency}

One consequence of Eq.~\eqref{eq:LOCCstrongestcontrapositive} is that each pair of LOCC-incomparable states is also LOSR-incomparable. We can leverage this observation, together with known results on LOCC-incomparability, to show that Eq.~\eqref{eq:multipartite_extension} in Corollary~\ref{cor:multipartite_extension} fails to be a sufficient condition for all $n > 2$.

Specifically, it is known that there exists chiral states for each $n > 2$ which are LOCC-incomparable to their complex conjugate (in a particular basis)~\cite{bennett2000exact,Kraus2010Local}. Eq.~\eqref{eq:LOCCstrongestcontrapositive} implies that these chiral states are also {\em LOSR-incomparable} to their complex conjugate. These instances of LOSR-incomparability are not witnessed by a failure of Eq.~\eqref{eq:multipartite_extension}
 because the squared Schmidt coefficients of any state are unchanged after complex conjugation. Consequently, satisfaction of Eq.~\eqref{eq:multipartite_extension} is not a sufficient condition for LOSR-convertibility.

For a concrete example with $n=3$,
  consider the tripartite state $\ket{\psi} = \ket{+++}_{ABC} + \frac{i - 1}{2 \sqrt{2}} \ket{111}_{ABC}$ and its complex conjugate (relative to the computational basis), $\ket{\psi^{\ast}}$. Evidently, the squared Schmidt coefficients of $\ket{\psi}$ and $\ket{\psi^{\ast}}$ 
 coincide for all bipartitions:
\begin{align} \label{psipsistar}
    \forall \beta : \lambda^{(\beta)}_{\psi} = \lambda^{(\beta)}_{\psi^{\ast}} = \left(\tfrac12 + \sqrt{\tfrac{5}{32}},\;\tfrac12 - \sqrt{\tfrac{5}{32}}\right).
\end{align}
It follows that the necessary condition for $\ket{\psi} \mapsto \ket{\psi^{\ast}}$ expressed in Eq.~\eqref{eq:multipartite_extension} is seen to be satisfied because one can take $\forall \beta: \lambda_{\zeta}^{(\beta)} = (1)$, which corresponds to $\ket{\zeta}$ being a product state,
  i.e., $\ket{\zeta} = \ket{\zeta_1}_A\ket{\zeta_2}_B\ket{\zeta_3}_C$. 

\section{A method for computing $\lambda_{\zeta}^{(\beta)}$} \label{alg}

This appendix delineates a method for computing $\lambda_{\zeta}^{(\beta)}$ in Eq.~\eqref{eq:multipartite_extension} (if it exists) for a fixed bipartition $\beta$ when $\ket{\psi}$ and $\ket{\phi}$ are known. This method is useful because, as previously mentioned, Eq.~\eqref{eq:multipartite_extension} constitutes a necessary constraint for LOSR (or LO) convertibility between $n$-partite pure states.
Moreover, in the case of bipartite systems, Eq.~\eqref{eq:multipartite_extension} reduces to the necessary and sufficient condition given by Eq.~\eqref{eq:bipartite_extension}, for which the following procedure also applies.

Importantly, Eq.~\eqref{eq:multipartite_extension} already constrains the Schmidt ranks of $\ket{\zeta}$ relative to all bipartitions,
\[ \forall \beta: \mathrm{SR}_{\psi}^{(\beta)} = \mathrm{SR}_{\phi}^{(\beta)} \mathrm{SR}_{\zeta}^{(\beta)}. \]

\begin{cor} 
    \label{rankcoro}
    A pure state $\ket \psi$ can be converted to pure state $\ket \phi$ by $\LOSR$ (or by $\LO$) only if for each bipartition $\beta$, the ratio of Schmidt ranks between $\ket \psi$ and $\ket \phi$ are positive integers, i.e.,
    \begin{align}
        \forall \beta : k^{(\beta)} \coloneqq \frac{\mathrm{SR}_{\psi}^{(\beta)}}{\mathrm{SR}_{\phi}^{(\beta)}} \in \mathbb Z^+.
    \end{align}
\end{cor}

Therefore, one should first compute $\lambda_{\psi}^{(\beta)}$ and $\lambda_{\phi}^{(\beta)}$ and thereby infer
$k^{(\beta)} \coloneqq \mathrm{SR}_{\psi}^{(\beta)} / \mathrm{SR}_{\phi}^{(\beta)}$;
if $k^{(\beta)}$
fails to be an integer for any bipartition $\beta$, then there is no solution to Eq.~\eqref{eq:multipartite_extension} for any $\lambda_{\zeta}^{(\beta)}$.
\rob Otherwise, $k^{(\beta)}$ will be an integer equal to the Schmidt rank with respect to the bipartition $\beta$ of $\ket{\zeta}$ (provided such a $\ket{\zeta}$ exists).

Without loss of generality, assume that the entries of $\lambda_{\psi}^{(\beta)}$ (and similarly $\lambda_{\phi}^{(\beta)}$ and $\lambda_{\zeta}^{(\beta)}$) are sorted in a non-increasing order, e.g. $\lambda_{\psi, 1}^{(\beta)} \geq \lambda_{\psi, 2}^{(\beta)} \geq \cdots \geq 0$.

Now for each $1 \leq j \leq \mathrm{SR}_{\zeta}^{(\beta)}$, let $\Lambda_{j}^{(\beta)}$ denote the multiset (where multiplicities are included)
\begin{align}
    \Lambda_{j}^{(\beta)} = (\lambda_{\phi,i}^{(\beta)}\lambda_{\zeta,j}^{(\beta)})_{i \in 1}^{\mathrm{SR}_{\phi}^{(\beta)}},
\end{align}
and let $\Lambda^{(\beta)}$ denote the multiset
\begin{align}
    \Lambda^{(\beta)} = (\lambda_{\psi,i}^{(\beta)})_{i \in 1}^{\mathrm{SR}_{\psi}^{(\beta)}}.
\end{align}
Notice that Eq.~\eqref{eq:multipartite_extension} is equivalent to the claim that the set of multisets $(\Lambda_{j}^{(\beta)})_{j = 1}^{\mathrm{SR}_{\zeta}^{(\beta)}}$ must form a multiset partition of $\Lambda^{(\beta)}$, i.e.
\begin{align}
    \label{eq:multiset_partition}
    \Lambda^{(\beta)} = \coprod_{j=1}^{\mathrm{SR}_{\zeta}^{(\beta)}}\Lambda_{j}^{(\beta)},
\end{align}
where the coproduct operation is the sum of multisets (generalizing the disjoint union of ordinary sets).

Without knowing the value of $\lambda_{\zeta,j}^{(\beta)}$, the value of $\Lambda_{j}^{(\beta)}$ is unknown. Fortunately, if a solution for $\lambda_{\zeta}^{(\beta)}$ in Eq.~\eqref{eq:multipartite_extension} exists, it must be that $\lambda_{\zeta,1}^{(\beta)}$ (the largest entry of $\lambda_{\zeta}^{(\beta)}$) is equal to
\begin{align}
    \label{eq:largest_value_trick}
    \lambda_{\zeta,1}^{(\beta)} = \frac{\lambda_{\psi,1}^{(\beta)}}{\lambda_{\phi,1}^{(\beta)}}.
\end{align}
Therefore, $\Lambda_{1}^{(\beta)}$ can be immediately determined. If $\Lambda_{1}^{(\beta)}$ is not contained in $\Lambda^{(\beta)}$, then no solution for $\lambda_{\zeta}^{(\beta)}$ in Eq.~\eqref{eq:multipartite_extension} exists and the procedure can be terminated. To compute the remaining entries of $\lambda_{\zeta}^{(\beta)}$, note that Eq.~\eqref{eq:multiset_partition} implies a generalization to Eq.~\eqref{eq:largest_value_trick}
\begin{align}
    \lambda_{\zeta,j}^{(\beta)} = \frac{1}{\lambda_{\phi,1}^{(\beta)}}\max\{ \lambda' | \lambda' \in \Lambda^{(\beta)} \setminus (\coprod_{i<j}\Lambda_{i}^{(\beta)}) \}.
\end{align}
If for any $j$, $\coprod_{i<j}\Lambda_{i}^{(\beta)}$ is not contained in $\Lambda^{(\beta)}$, then again no solution for $\lambda_{\zeta}^{(\beta)}$ exists. Otherwise, the computed values for $\lambda_{\zeta,j}^{(\beta)}$ as $j$ ranges from $1$ through $\mathrm{SR}_{\zeta}^{(\beta)}$ will constitute the unique solution for $\lambda_{\zeta}^{(\beta)}$ in Eq.~\eqref{eq:multiset_partition}.
\blk

\section{Bipartite entanglement catalysis under LOSR}\label{catalysis}

Recall that for bipartite pure states, the condition for $\ket \psi \mapsto \ket \phi$ under LOSR can be expressed in terms of vectors of squared Schmidt coefficients as
\begin{equation}\label{AppE1}
\exists \ket \zeta: \lambda_{\psi}^{\downarrow} = (\lambda_{\phi} \otimes \lambda_{\zeta})^{\downarrow}.
\end{equation}
(See Eq.~\eqref{eq:bipartite_extension},  and recall that  $v^{\downarrow}$ denotes the vector with the same components as $v$, but organized in non-increasing order, $v^{\downarrow}_1 \ge v^{\downarrow}_2  \ge \dots \ge v^{\downarrow}_n $.\blk )
Consequently, the condition for $\ket \psi \otimes \ket \chi \mapsto \ket \phi \otimes \ket \chi$ under LOSR is 
\begin{equation}\label{AppE2}
\exists \ket \zeta: (\lambda_{\psi}\otimes \lambda_{\chi})^{\downarrow} = (\lambda_{\phi} \otimes \lambda_{\chi} \otimes \lambda_{\zeta})^{\downarrow}.
\end{equation}  
\rob We here demonstrate that if a particular state $\ket \zeta$ satisfies the equality of Eq.~\eqref{AppE2} for some catalyst $\ket \chi$, then the same state $\ket \zeta$ also satisfies the equality of Eq.~\eqref{AppE1}, in the absence of any catalyst.

For the purposes of making the proof more transparent, we introduce a simplified notation for the vectors appearing therein: $\alpha \equiv \lambda_{\phi} \otimes \lambda_{\zeta}$,  $\beta \equiv \lambda_{\psi}$ and $\chi \equiv \lambda_{\chi}$.  In this notation, what we seek to prove is that  
\begin{equation}\label{AppE3}
(\chi \otimes \alpha)^{\downarrow} = (\chi \otimes \beta)^{\downarrow}
\end{equation}
 implies 
 \begin{equation}\label{AppE4}
 \alpha^{\downarrow} = \beta^{\downarrow}.
 \end{equation}

Note first of all that if one can find a catalyst state $\ket \chi$ that is characterized by a vector of squared Schmidt coefficients, $\chi$, that has some components that are zero, then one can also achieve catalysis by a state with no such components, simply by defining the state on a Hilbert space of smaller dimension. Consequently, one can assume without loss of generality that $\chi_i \ne 0$ for all $i$. 

The proof is by induction.  

Clearly, $(\chi \otimes \alpha)^{\downarrow}_1 = \chi^{\downarrow}_1 \alpha^{\downarrow}_1$ and $(\chi \otimes \beta)^{\downarrow}_1 = \chi^{\downarrow}_1 \beta^{\downarrow}_1$.  Using  $(\chi \otimes \alpha)^{\downarrow}_1 = (\chi \otimes \beta)^{\downarrow}_1$, i.e., the first component of the vector equality expressed by Eq.~\eqref{AppE3}, 
 and recalling that $\chi^{\downarrow}_1 \ne 0$, we infer that $\alpha^{\downarrow}_1 =\beta^{\downarrow}_1$.

Let  $\alpha^{(n)}$ denote the vector obtained by dropping the first $n$ components of $\alpha^{\downarrow}$, so that $\alpha^{(n)} \equiv (\alpha^{\downarrow}_{n+1},\alpha^{\downarrow}_{n+2},\dots)$.  Define $\beta^{(n)}$ similarly.  

Note that $\alpha^{\downarrow}_1 =\beta^{\downarrow}_1$ implies $ \chi^{\downarrow}_i \alpha^{\downarrow}_1 =  \chi^{\downarrow}_i \beta^{\downarrow}_1$ for all $i$.  If we drop from the vector $\chi \otimes \alpha$ the components $\chi^{\downarrow}_1 \alpha^{\downarrow}_1, \chi^{\downarrow}_2 \alpha^{\downarrow}_1, \dots, \chi^{\downarrow}_n \alpha^{\downarrow}_1 $, and put the remaining components of $\chi \otimes \alpha$ in non-increasing order, we obtain the vector $(\chi \otimes \alpha^{(1)})^{\downarrow}$.  The analogous procedure on the vector $\chi \otimes \beta$ yields the vector $(\chi \otimes \beta^{(1)})^{\downarrow}$.  Now note that the components we have dropped from $\chi \otimes \alpha$ are equal to the corresponding components that we have dropped from $\chi \otimes \beta$, and so Eq.~\eqref{AppE3} implies that the remaining components of these vectors, when placed in nonincreasing order, are equal to one another, i.e.,
\begin{equation}\label{AppE6}
(\chi \otimes \alpha^{(1)})^{\downarrow}= (\chi \otimes \beta^{(1)})^{\downarrow}.
\end{equation}

Clearly, the argument provided above to justify the inference from Eq.~\eqref{AppE3} to $\alpha^{\downarrow}_1 =\beta^{\downarrow}_1$ also justifies an inference from Eq.~\eqref{AppE6} to $\alpha^{\downarrow}_2 =\beta^{\downarrow}_2$.

Similarly, the argument provided above to justify the inference from $\alpha^{\downarrow}_1 =\beta^{\downarrow}_1$ and Eq.~\eqref{AppE3} to Eq.~\eqref{AppE6} also justifies an inference from $\alpha^{\downarrow}_2 =\beta^{\downarrow}_2$ and Eq.~\eqref{AppE3} to 
\begin{equation}\label{AppE8}
(\chi \otimes \alpha^{(2)})^{\downarrow}= (\chi \otimes \beta^{(2)})^{\downarrow}.
\end{equation}

By repeating this sequence of arguments inductively, one can infer that $\alpha^{\downarrow}_i =\beta^{\downarrow}_i$ for all $i$, which concludes the proof. 

\section{\sloppy Example of a pair of network structures for which the classical-nonclassical boundary for states differs from the separable-nonseparable boundary}\label{ExamplePairStructures}

Consider the two networks among three parties depicted in Fig.~\ref{fig:trine_and_triangle_networks}. The first, termed the {\em triangle network}, has three 
sources, each of which is shared between a different pair of parties (these sources are therefore termed ``2-way sources'').  The second,
 which we term the {\em  trine network}, 
has a single 
 source shared between all three parties (which is termed a ``3-way source'').  
\begin{figure}[htb]
\centering
    \subfloat[The triangle network.]{
        \makebox[0.24\textwidth][c]{
            \raisebox{-0.5\height}{\includegraphics[scale=1]{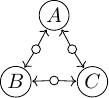}}
        }
        \label{fig:triangle_network}
    }
    \subfloat[The trine network.]{
        \makebox[0.24\textwidth][c]{
            \raisebox{-0.5\height}{\includegraphics[scale=1]{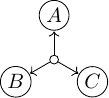}}
        }
        \label{fig:trine_network}
    }
    \caption{The triangle network and the trine network.}
    \label{fig:trine_and_triangle_networks}
\end{figure}

\begin{figure}[htb]
\centering
    \subfloat[Form of a tripartite state that is realizable in a classical version of the trine network; that is, by LOSR operations.]{
        \makebox[0.48\textwidth][c]{
            \raisebox{-0.5\height}{\includegraphics[scale=1.5]{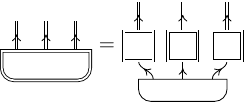}}
        }
        \label{fig:trine_losr}
    } \\
    \subfloat[Form of a tripartite state that is realizable in a quantum version of the trine network; that is, by LOSE operations.]{
        \makebox[0.48\textwidth][c]{
            \raisebox{-0.5\height}{\includegraphics[scale=1.5]{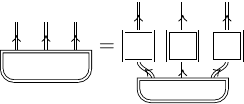}}
        }
        \label{fig:trine_lose}
    }
    \caption{A distinction among tripartite states that can be realized in the trine network.}
    \label{fig:trine_comparison}
\end{figure}

\begin{figure}[htb]
\centering
    \subfloat[Form of a tripartite state that is realizable in a classical version of the triangle network, that is by LO2WSR.]{
        \makebox[0.48\textwidth][c]{
            \raisebox{-0.5\height}{\includegraphics[scale=1.5]{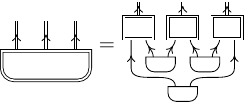}}
        }
        \label{fig:triangle_losr}
    } \\
    \subfloat[Form of a tripartite state that is realizable in a quantum version of the triangle network, that is by LO2WSE.]{
        \makebox[0.48\textwidth][c]{
            \raisebox{-0.5\height}{\includegraphics[scale=1.5]{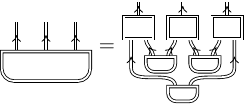}}
        }
        \label{fig:triangle_lose}
    }
    \caption{A distinction among tripartite states that can be realized in the triangle network.}
    \label{fig:triangle_comparison}
\end{figure}

Consider the tripartite states that can be realized in the trine network.    Among such states, the boundary between those having classical correlational properties and those having nonclassical correlational properties is determined by whether the 3-way source can be classical or whether it needs to be quantum.
More precisely, it is the boundary between tripartite states that can be achieved using local operations together with a {\em classical} 3-way source, i.e., shared randomness between the three parties,  termed LOSR and depicted in Fig.~\ref{fig:trine_losr}, and those that can only be achieved using local operations together with a {\em quantum} 3-way source, i.e., shared entanglement between the three parties,  termed LOSE~\cite{Gutoski2008,Schmid2021postquantumcommon} and depicted in Fig.~\ref{fig:trine_lose}.
As is well known, LOSE allows for the realization of {\em any} tripartite state, and the boundary between those states that are realizable by LOSR and those that require LOSE is simply 
 the boundary between the separable and the nonseparable states.
 
Contrast this with the tripartite states that can be realized in the {\em triangle} network.    Among such states, the boundary between those having classical correlational properties and those having nonclassical correlational properties is determined by whether the triple of 2-way sources can be classical or whether one or more of them needs to be quantum.
More precisely, it is the boundary between tripartite states 
that can be achieved using  Local Operations together with sources of 2-Way Shared Randomness (LO2WSR), depicted in Fig.~\ref{fig:triangle_losr}, and those which can only be achieved with Local Operations together with sources of 2-Way Shared Entanglement (LO2WSE), 
depicted in Fig.~\ref{fig:triangle_lose}.\footnote{Note that both LO2WSR and LO2WSE are distinct from 
 Local Operations together with Shared Randomness  and 2-Way sources of Shared Entanglement (LOSR2WSE) discussed in Sec.~\ref{sec:selftesting}.}   
The set of tripartite states that are realizable in a triangle network does {\em not} include all tripartite states.  Indeed, there are both nonseparable and separable states that cannot be realized in such a network~\cite{navascues2020genuine}. 
Furthermore, the boundary between states realizable in this network classically versus those realizable nonclassically does not coincide with the boundary between separable and nonseparable states.

To see that the latter claim is true,
 it is sufficient to consider a simple example, corresponding to the tripartite probability distribution proposed by Fritz~\cite{Fritz2012beyondBell}, but conceptualized as a separable state on a tripartite system, $ABC$, where each subsystem is associated with a 4-dimensional Hilbert space.
  Defining binary variables $A_0$,$A_1$, $B_0$,$B_1$ and $C_0$,$C_1$, this tripartite separable state
 can be expressed as:
\begin{align}\label{Fritzstate}
\rho^{\rm Fritz}_{ABC} &=\sum_{a_0,a_1,b_0,b_1,c_0,c_1} P^{\rm Fritz}_{A_0 A_1 B_0 B_1 C_0 C_1}(a_0 a_1 b_0 b_1 c_0 c_1)\nonumber\\
&\times  | a_0 a_1 \rangle_A \langle a_0 a_1 | \otimes  | b_0 b_1 \rangle_B \langle b_0 b_1 |\otimes | c_0 c_1 \rangle_C \langle c_0 c_1 |,
\end{align}
where $P^{\rm Fritz}_{A_0 A_1 B_0 B_1 C_0 C_1}$ is the tripartite distribution proposed by aseFritz~\cite{Fritz2012beyondBell}:
\begin{align}\label{Fritzdistn}
P^{\rm Fritz}_{A_0 A_1 B_0 B_1 C_0 C_1} = P_{B_1 C_1} P_{A_1 C_0} P_{A_0 B_0 |A_1 B_1},
\end{align}
with
\begin{align}
P_{B_1 C_1} =  (\tfrac{1}{2} \delta_{C_1,0} +\tfrac{1}{2} \delta_{C_1,1}) \delta_{B_1, C_1}\nonumber\\
P_{A_1 C_0} =  (\tfrac{1}{2} \delta_{C_0,0} +\tfrac{1}{2} \delta_{C_0,1}) \delta_{A_1, C_0}\nonumber
\end{align}
and where $P_{A_0 B_0 |A_1 B_1}$ describes a conditional probability distribution that provides the maximum quantum violation of the CHSH inequalities, termed a {\em Tsirelson box}.

By virtue of being separable, this state is realizable classically in the trine network.
 It suffices to use the 3-way source of shared randomness.\footnote{Of source, it is also realizable {\em quantumly} in the 
 trine network since a 3-way source of shared entanglement is strictly more powerful than a 3-way source of shared randomness.} 
 However, in the triangle network, it is {\em only} realizable quantumly and not classically,  That is, given only access to 2-way sources, this state can not be realized if the sources are classical, i.e., if each provides only shared randomness, but it {\em can} be realized if they are quantum, i.e., if they provide shared entanglement.\footnote{We are here restricting attention to the distinction between {\em all} sources being classical and its converse.
The state in question can in fact be realized if the 2-way source between $A$ and $B$ is quantum, while the other pair of 2-way sources are classical. 
}
   The latter claim follows from the results of Ref.~\cite{Fritz2012beyondBell}.  
(It suffices to note that if this separable state {\em could} be classically realized in the triangle network, then by local measurements in the bases in which it is diagonal, one could generate the tripartite probability distribution that is proven in Ref.~\cite{Fritz2012beyondBell} to be not classically realizable in the triangle network.) 
Consequently, for the triangle network, there are separable states, such as the one in Eq.~\eqref{Fritzstate}, that exhibit nonclassical correlational properties insofar as they cannot be realized if the sources in the network are restricted to be of the classical variety.  This proves that the boundary between classical and nonclassical states in the triangle network is {\em not} the boundary between separable and nonseparable states.\footnote{The result can also be proven using the machinery of Ref.~\cite{navascues2020genuine}.}

It is worth noting that the way in which the nonclassicality of the correlational properties of states are {\em quantified} will also differ between the trine and triangle network structures. The free operations in the trine network are LOSR, while the free operations in the triangle network are LO2WSR.
 Since LO2WSR is strictly contained within LOSR, there can be pairs of states that are strictly ordered in the trine network while they are incomparable in the triangle network.\footnote{Although we here advocate for referring to tripartite states that are achievable by LOSE but not by LOSR 
     (i.e., achievable by a circuit of the form of Fig.~\ref{fig:trine_lose}, but not by a circuit of the form of Fig.~\ref{fig:trine_losr})
     as {\em LOSR-entangled}, and the states that are achievable by LOSE but not by LOSR2WSE (i.e., achievable by a circuit of the form of Fig.~\ref{fig:trine_lose}, but not by a circuit of the form of Fig.~\ref{genuinely3way}) as {\em genuine 3-way entangled}, we  do {\em not} advocate for referring to states that are achievable by LO2WSE but not by LO2WSR  (i.e., achievable by a circuit of the form of Fig.~\ref{fig:triangle_lose}, but not by a circuit of the form of Fig.~\ref{fig:triangle_losr})
  as ``entangled''.
  Although it may at first glance seem natural to define a state to be entangled in a given network if and only if it is nonfree relative to having the communication channels and sources in that network being classical,
   such a terminology would conflict with the standard convention of referring to all separable states as unentangled.  As the latter convention is entrenched, it is not advisable to try and overturn it. Rather, it makes more sense to refer to states that are nonfree relative to LO2WSR
  simply as {\em quantumly correlated relative to LO2WSR}.  The sorts of networks for which all of the nonclassical states will be nonseparable are those for which it is possible, using classical sources, to realize {\em any} joint probability distribution over the parties' outputs. }

Similar considerations hold for the classical-nonclassical distinction for correlational properties of boxes.  In a tripartite network with a common source (i.e., a tripartite Bell scenario), the boundary between boxes having classical correlational properties and those having nonclassical correlational properties is the boundary between those that satisfy all the Bell inequalities and those that violate some Bell inequality.  In the triangle network, however, the distinction is picked out by a different set of inequality constraints.  A prescription for finding all of these 
 can be given in terms of the inflation technique for causal inference~\cite{Wolfe2016inflation,WolfeNavascues}. To determine the order over boxes, 
   one again uses LOSR in the trine network,
    but LO2WSR in the triangle network.\footnote{In a Bell scenario, i.e., a network with a common source shared by all the parties, boxes that can only be realized by quantum sources are conventionally termed ``nonlocal''.  In Footnote~\ref{footnote:nolocality}, we noted that this terminology is not particularly good for describing the nonclassicality of correlational properties of boxes.  It is even less suited to describing the nonclassicality of correlational properties of boxes in alternatives to the Bell scenario, such as the triangle network.  
    } 
    
\end{appendices}
\end{document}